\documentclass[journal]{IEEEtran}

\usepackage{ifthen}
\usepackage{subfigure}
\usepackage[final]{graphicx}

\usepackage{amsthm}
\usepackage{cite}
\usepackage[cmex10]{amsmath}
\usepackage{amsfonts}
\usepackage{array}
\usepackage{booktabs}
\usepackage{colortbl}

\usepackage{mdwtab}
\usepackage{eqparbox}

\usepackage[ruled,vlined]{algorithm2e}

\newtheorem{proposition}{Proposition}
\newtheorem{lemma}{Lemma}

\newtheorem{remark}{Remark}

\usepackage{array}
\usepackage{booktabs}
\usepackage{colortbl}
\usepackage{mdwmath}
\usepackage{mdwtab}
\usepackage{eqparbox}

\begin{document}
\title{Analytical Modeling of Interference Aware Power Control for the Uplink of Heterogeneous Cellular Networks}
%

\author{Francisco~J.~Martin-Vega, Gerardo~Gomez, 
        Mari~Carmen~Aguayo-Torres and Marco~Di~Renzo
\thanks{This work has been submitted to the IEEE for possible publication. Copyright may be transferred without notice, after which this version may no longer be accessible} 
\thanks{F.~J.~Martin-Vega, G. Gomez and M. C. Aguayo-Torres are with the Departamento 
de Ingenier\'ia de Comunicaciones, Universidad de M\'alaga, M\'alaga 29071, Spain (e-mail: fjmvega@ic.uma.es, ggomez@ic.uma.es, aguayo@ic.uma.es)}
\thanks{Marco Di Renzo is with Laboratoire des Signaux et Syst\`emes, Centre National
de la Recherche Scientifique-\'Ecole Supérieure d’\'Electricit\'e-Universit\'e ParisSud XI, 91192 Gif-sur-Yvette Cedex, France (e-mail: marco.direnzo@lss.supelec.fr).}}

\maketitle


\newif\ifOneColumn
\OneColumnfalse

\begin{abstract}
Inter-cell interference is one of the main limiting factors in current Heterogeneous Cellular Networks (HCNs). Uplink Fractional Power Control (FPC) is a well known method that aims to cope with such limiting factor as well as to save battery live. In order to do that, the path losses associated with Mobile Terminal (MT) transmissions are partially compensated so that a lower interference is leaked towards neighboring cells. Classical FPC techniques only consider a set of parameters that depends on the own MT transmission, like desired received power at the Base Station (BS) or the path loss between the MT and its serving BS, among others. Contrary to classical FPC, in this paper we use stochastic geometry to analyze a power control mechanism that keeps the interference generated by each MT under a given threshold. We also consider a maximum transmitted power and a partial compensation of the path loss. Interestingly, our analysis reveals that such Interference Aware (IA) method can reduce the average power consumption and increase the average spectral efficiency. Additionally, the variance of the interference is reduced, thus improving the performance of Adaptive Modulation and Coding (AMC) since the interference can be better estimated at the MT.
\end{abstract}


%
\IEEEpeerreviewmaketitle

\section{Introduction}
\label{sec:Introduction}
%
\IEEEPARstart{S}{ince} the beginning of cellular systems, interference has been the main limiting factor, due in part to its highly indeterministic nature and its sensitivity to network conditions. This situation is even aggravated in the uplink, since the interfering set of Mobile Terminals (MTs) depends on the scheduling decisions of other cells (that may change from one sub-frame to the following), their channel states, positions and transmission powers. Additionally, in irregular networks interfering MTs can be closer to the serving BS than its intended MT. 

In order to cope with interference, Fractional Power Control (FPC) has been included as an essential part of the uplink (UL) of Long Term Evolution (LTE) and LTE-Advanced \cite{3gpp.36.213}. %
Essentially, this technique partially compensates the path loss allowing cell-interior MTs to save battery while ensuring that cell-edge MTs do not cause excessive interference to neighboring cells \cite{Castellanos08, Mullner09, Simonsson08}. However such approach still generates an undesired level of interference that reduces the UL performance. Beside this, the highly indeterministic nature of UL interference poses additional challenges to interference estimation which degrades the performance of Adaptive Modulation and Coding (AMC) \cite{Goldsmith98}. 

To solve this issue, an Interference Aware Fractional Power Control (IAFPC), which is compliant with LTE specifications, was proposed in \cite{Zhang12}. This power control mechanism establishes a maximum interference level $i_0$ that each MT transmission can cause to the most interfered Base Station (BS). Interestingly, this method leads to a significant performance improvement regarding the variance of the interference, average rate and power consumption. There are several works that have studied IA power control methods \cite{Zhang12, Rao07, Boussif08}; however they are based on simulations with a single tier and an hexagonal grid. 
Besides simulation based studies, there is a need to obtain analytical models that allow for tractable analysis leading to a better understanding of both power control and association in the UL. In addition, analytical models allow for quick evaluation of main performance metrics and optimization.

\subsection{Related Work}
\label{sec:Related Work}

Traditionally, the association between MTs and BSs in the UL has been coupled with the association in downlink (DL) for the sake of technical reasons related to network implementation. However, it has been  recently proposed to split the UL and DL association, making the association in the UL based on a minimum path loss criterion \cite{Elshaer14}. This approach also allows to reduce UL interference since MTs are associated with the minimum path loss BSs, thus MTs will transmit with smaller power.

The variant demand of capacity across service areas and the deployment of Heterogeneous Cellular Networks (HCNs) are evolving the cellular network from a regular grid to a rather irregular infrastructure which resembles more to a random topology. In this context, stochastic geometry appears as an interesting tool that allows for a tractable analysis of cellular systems where the positions of the BSs typically follows the uniform Poison Point Process (PPP) \cite{Haenggi09,ElSawy13}.  There are several works illustrating that stochastic geometry offers lower bounds in performance which are as tight as simulation results obtained with the hexagonal grid \cite{Andrews11, Dhillon12, Guo13, Blaszczyszyn13, Xu11}. Regarding analysis, this approach normally considers the typical link between a probe MT and its serving BS, where the term \textit{typical}, means randomly chosen and involves that the MT can be placed anywhere inside the cell. 

Analysis of the UL is quite more involving than the DL; the transmission powers of the interferers are coupled with their serving BSs' distances due to power control. In addition, even if the positions of the BSs and MTs follow a PPP, the positions of the interfering MTs do not follow a PPP, making the exact analysis intractable \cite{Singh15}. The reason behind that is related to the fact that in 4G networks orthogonal resource allocation is used, and hence there will be a single interfering MT per BS. 

There are several works that analyze the performance of UL with different power control and association policies. In \cite{Novlan13} single tier networks with FPC is analysed using stochastic geometry. In order to avoid the intractability of the interfering MT positions, the proposed framework assumes that MTs scheduled in the Resource Block (RB) of interest form a Voronoi Tesselation and a single BS falls inside each Voronoi cell. In \cite{Smiljkovikj15} FPC is analyzed approximating the positions of the interfering MTs as an uniform PPP in the entire plane; however, with this approach, some interfering MTs might experience a lower path loss with the probe BS than with its serving BS, which is not a realistic situation. Recently \cite{ElSawy14, Singh15} propose accurate frameworks to model the interfering MT positions. 
Both works consider the correlation between the probe BS and the interfering MT positions. Such correlation involves that interfering MTs cannot be placed in positions that would result associated with the probe BS. 
The authors of \cite{ElSawy14} consider a truncated channel inversion power control, where MTs tries to fully compensate their path loss provided that they do not have to transmit with more power than $p_\mathrm{max}$. Those MTs requiring a higher power than $p_\mathrm{max}$ are kept silent and association is based on min path loss. A generalized weighted association for both coupled and decoupled access is considered in \cite{Singh15} so as to analyze FPC. Then the joint rate in DL/UL is obtained to assess the trade-off between DL and UL performance. 

\subsection{Contributions}
\label{sec:Contributions}
In this paper, a novel framework for the modeling and analysis of IAFPC is presented. This model is appealing since non IA power control methods can be viewed as a particular case when the maximum allowed interference level $i_0$ tends to infinity. 
The proposed framework is based on a \textit{conditional thinning} in order to appropriately model the positions of the interfering MTs, adding a necessary correlation with the probe BS position. However, this case is more involved than the case of previously studied FPC since the correlation with the most interfered BS also needs to be considered. In addition, it is necessary to deal with non linear functions that depends on both distance towards the serving BS and the distance towards the most interfered BS. All these issues make that final expressions for the distribution of the Signal to Interference Plus Noise Ratio (SINR) complex to evaluate. Hence, two approximations to the interference are proposed aiming to reduce the computational complexity. These approximations are: (i) approximate the Laplace transform of the interference by a sigmoid function and then perform logistic regression \cite{Cybenko89} in order to obtain the function parameters and (ii) approximate the Laplace transform by a suitable function and then perform Moment Matching (MM). Additionally, asymptotic analysis is performed in order to identify trends as $i_0$ becomes low or tends to infinity. This latter approach also reduces the numerical complexity avoiding the need to approximate the interference term. The contributions of the paper can be summarized in the following points:
\begin{itemize}
\item Analysis and comparison of both IA and non IA FPC. From numerical comparison it is shown that interference awareness reduces both the mean and variance of the interference as well as the average transmitted power and increases the average Spectral Efficiency (SE).
\item Analysis of a wide variety of performance metrics in order to gain theoretical comprehension of the aforementioned techniques. Considered key performance metrics include average transmitted power, mean and variance of the interference, coverage probability and average spectral efficiency.
%
%
\item Two accurate approximations to the Laplace transform of the interference for the case of IAFPC are proposed. These approximations efficiently reduce the computational complexity. 
\item Asymptotic analysis reveals interesting insights from theoretical expressions. In particular it is shown that statistics of the interference are independent of the BS density in the low $i_0$ regime under minimum path loss association. 
\end{itemize}

\subsection{Paper Organization and Notations}
\label{sec:Paper Organization and Notations}
The rest of the paper is organized as follows. Section \ref{sec:System model} describes the system model and proposed approach. In Section \ref{sec:Analysis of Interference Aware Power Control} the analysis of IAFPC  is presented. Two necessary approximations to the interference for IAFPC are proposed in \ref{sec:Statistical Modeling of the interference}. Section \ref{sec:Asymptotic Analysis} presents asymptotic analysis of IAFPC yielding simpler expressions that represent the performance of non IA FPC when $i_0$ tends to $\infty$ and the performance of IAFPC in the low $i_0$ regime. Then, in Section \ref{sec:Numerical Results} the derived expressions are evaluated to obtain insights and identify trends. Finally, concluding remarks are given Section \ref{sec:Conclusion}. 

\textbf{Notation}: Through this paper $\mathbb{E}[\cdot]$ stands for the expectation operator and $\Pr(\cdot)$ for the probability measure. $\mathbf{1} (\cdot)$ is the indicator function. The first and second derivatives of $f(x)$ evaluated at $x_0$ are represented as $f'(x_0)$ and $f''(x_0)$. Random variables (RV) and events are represented with capital letters whereas lower case is reserved for deterministic values and parameters. If $X$ is a RV, $f_X(\cdot)$, $F_X(\cdot)$, $\bar{F}_X(\cdot)$ and $\mathcal{L}_X(\cdot)$ represent its probability density function (pdf), cumulative distribution function (cdf), complementary cdf (ccdf) and Laplace transform of its pdf respectively. $\Gamma(z)=\int_{0}^{\infty} t^{z-1} \mathrm{e}^{-t} \mathrm{d}t$ stands for the Euler gamma function whereas 
$_2F_1(\cdot,\cdot,\cdot,\cdot)$ is the Gauss hypergeometric function defined in \cite{Abramowitz65} (Ch. 15). Having a function $f(x,y)$ we write the limit when $x \to a$ as $f(y)^{(x \to a)} = \lim \limits_{x \to a} f(x,y)$. Finally we say that two functions $f(x)$ and $f^{(x \sim a)}(x)$ are asymptotically similar when $x \to a$ if 
$\lim \limits_{x \to a} \frac{f(x)}{f^{(x \sim a)}(x)} = 1$.

\section{System Model}
\label{sec:System model}
We consider a HCN composed of two tiers, i.e. Macro cell BSs (MBSs) and Small cell (BSs) SBSs, where the BSs of tier 
$j\in\mathcal{K}=\{1,2\} $ are spatially distributed in $\mathbb{R}^2$ according to an uniform PPP $\Phi^{(j)}=\{\mathrm{BS}^{(j)}_0, \mathrm{BS}^{(j)}_1, \cdots \}$, with density $\lambda^{(j)}$ where $\mathrm{BS}^{(j)}_i$ is the location of $i$th BS in the $j$th tier.  The positions of all BSs are represented with the PPP $\Phi=\cup_{j\in \mathcal{K}} \Phi^{(j)}$. MTs are also spatially distributed according to an uniform PPP, $\Phi_{\mathrm{MT}}=\{\mathrm{MT}_0,\mathrm{MT}_1,\cdots\}$, with density $\lambda_{\mathrm{MT}}$. It is assumed that the density of MTs is high enough to consider full loaded conditions, i.e. each BS has at least one MT to serve and all the RBs are used. The analysis is performed for the typical MT, i.e. a randomly chosen MT with $\mathrm{MT}\in\Phi_{\mathrm{MT}}$. Since uniform PPPs are translation invariant point processes, such typical point can be considered to be placed at the origin (without loss of generality) thanks to Slivnyak's theorem \cite{Haenggi13}. In this paper the typical MT is named the probe MT, which is represented as $\mathrm{MT}_0$, and its serving BS is the probe BS,  $\mathrm{BS}_0$. 

\subsection{Channel Models}
\label{sec:Channel Models}
It is considered that transmitted signal undergoes both shadowing and multi-path fading. Multi-path fading is modeled as an exponential distribution with unitary mean whereas shadowing is modeled as a Log-normal distribution with standard deviation $\sigma_s$ and mean $\mu_s$. Unitary mean squared value is assumed and hence $\mu_s= -\ln(10) \sigma_s^2/20$ \cite{DiRenzo13}. We use equivalent distances including shadowing as in \cite{Dhillon14} where  $\dot{R}_{x,y}={S}^{-1/\alpha}_{x,y} R_{x,y}$ being ${S}_{x,y}$ the shadowing between locations $x$ and $y$ and $R_{x,y}=\| x-y \|$ the Euclidean distance. It is assumed independent fading and shadowing for different locations. 
Hence shadowing can be considered as a random displacement over $\Phi^{(j)}$  \cite{Blaszczyszyn13,Haenggi13}, where the density of the displaced PPP is $\dot{\lambda}^{(j)}=\lambda^{(j)} \mathbb{E} [S^{-1/\alpha}]$ and $\mathbb{E} [S^{-1/\alpha}] = \exp \left( \frac{\ln(10) \mu_s}{5\alpha} + \frac{1}{2} \left( \frac{\ln(10) \sigma_s}{5 \alpha}\right) \right)$.  
For the sake of simplicity, from now on all the distances and PPPs are considered to include shadowing. 

Path loss models proposed in 3GPP for performance evaluation are typically formulated as 
$L(\mathrm{dB})=a_L+b_L \log_{10}(R_{x,y}(\mathrm{km}))$ where $R_{x,y}$ is the distance between locations $x$ and $y$ expressed in km and both $a_L$ and $b_L$ depends on several radio frequency parameters. However, in the literature theoretical analysis normally considers a path loss law that only depends on the path loss exponent $\alpha$ as 
$L=\left( R_{x,y}(\mathrm{m})\right)^{\alpha}$. In this work we consider a path loss law as in \cite{Blaszczyszyn13} that considers also a path loss slope $\tau$. Hence the path loss is expressed as $L=\left(\tau \cdot R_{x,y}(m)\right)^\alpha$. This approach allows to include a more realistic path loss models since $\tau=10^{(a_L-3b_L)/b_L}$ and $\alpha=b_L/10$. In particular, we have considered in Section \ref{sec:Numerical Results} the path loss model from 3GPP (\cite{3gpp.136.942}, Sec. 4.5.2) that appears below for convenience:
\ifOneColumn
\begin{equation}
\label{eq:3GPP path loss}
     a_L = 80 - 18 \log_{10}\left( h_\mathrm{BS}(\mathrm{m}) \right)
     + 21 \log_{10}\left( f_c(\mathrm{MHz}) \right); \, 
     b_L = 40 \left( 1-4 \cdot 10^{-3} h_\mathrm{BS}(\mathrm{m}) \right)   
\end{equation}
\else
\begin{equation}
\label{eq:3GPP path loss}
  \left\{ \begin{array}{lr} 
     a_L = 80 - 18 \log_{10}\left( h_\mathrm{BS}(\mathrm{m}) \right)
     + 21 \log_{10}\left( f_c(\mathrm{MHz}) \right) \\ 
     b_L = 40 \left( 1-4 \cdot 10^{-3} h_\mathrm{BS}(\mathrm{m}) \right) \end{array}  \right.
\end{equation}
\fi
\noindent 
where $f_c$ is the carrier frequency in MHz and $h_\mathrm{BS}(\mathrm{m})$ is the BS's antenna height in m. 

\subsection{Association and Scheduling}
\label{sec:Association}
The criterion of association among MTs and BSs is based on average weighted received power as in \cite{Singh15} using association weights $t^{(j)}$ for tier $j\in\mathcal{K}$. 

Let us define the event $\mathcal{X}^{(j)}_{\mathrm{MT}_i}$ as: $\mathrm{MT}_i$ \textit{is associated with tier} $j$. More formally such event can be defined as follows
%
\begin{equation}
\label{eq:Xj}
\mathcal{X}^{(j)}_{\mathrm{MT}_i} = 
	\left\{ {t^{(j)}}{\left( \tau \cdot R^{(j)}_{\mathrm{MT}_i,(1)}
	 \right)^{-\alpha}}   
	>  {t^{(\tilde{j})}}{\left( \tau \cdot R^{(\tilde{j})}_{\mathrm{MT}_i,(1)} 
	\right)^{-\alpha}}   \right\}
\end{equation}
\noindent where $\tilde{j} = \left\{x \in \mathcal{K}: x \neq j \right\}$ represents the complementary tier to tier $j$ and ${R}^{(\tilde{j})}_{x,(q)}$ is the distance from x to the $q$th nearest BS of tier $\tilde{j}$, i.e. ${R}^{(\tilde{j})}_{x,(1)}$ is the distance to the nearest BS. It can be noticed that the association weights allow us to model minimum path loss association with $t^{(j)}=1$ and association based on DL received power among others. 

Full frequency reuse is considered where all the BSs share the same bandwidth which is divided in RBs for scheduling purposes, being a single RB the minimum amount of bandwidth that can be allocated. In LTE and LTE-A each RB is divided into 12 Resource Elements (REs) of 15 kHz. 
All MTs that have been scheduled in a given resource are name \emph{active} MTs. 
The set of active interfering MTs from tier $k$ scheduled in the probe RB is identified as $\Psi^{(k)}$. The MT scheduled in each RB is selected randomly. 

\subsection{Power Control Mechanism}
\label{sec:Power Control Mechanism}
In this work we consider an IAFPC mechanism as in \cite{Zhang12} where each MT causes less interference than $i_0$ to its most interfered BS and transmits with less power than $p_\mathrm{max}$. It is assumed that power control can adapt to slow variation in received power, hence it can only compensate for path loss and shadowing. The transmit power can be expressed as follows
\ifOneColumn
\begin{align}
\label{eq:pMT Soft}
	p_\mathrm{MT} \left( R_{\mathrm{MT}_0}, U_{\mathrm{MT}_0} \right) = 
	\min  \left(  p_0 \left( \tau R_{\mathrm{MT}_0} \right)^{\alpha \epsilon}, 
	i_0 \left( \tau U_{\mathrm{MT}_0} \right)^\alpha, p_\mathrm{max} \right)
\end{align}
\else
\begin{align}
\label{eq:pMT Soft}
	p_\mathrm{MT} \left( R_{\mathrm{MT}_0}, U_{\mathrm{MT}_0} \right) = 
	\min & \left(  p_0 \left( \tau R_{\mathrm{MT}_0} \right)^{\alpha \epsilon}, \right. \nonumber \\
	& i_0 \left. \left( \tau U_{\mathrm{MT}_0} \right)^\alpha, p_\mathrm{max} \right)
\end{align}
\fi
\noindent
being $R_{\mathrm{MT}_0}$ the distance to the serving BS, $U_{\mathrm{MT}_0}$ the distance to the most interfered BS, $p_0$ is the desired received signal power at the serving BS and $\epsilon$ the partial compensation factor. If the transmit power $p_\mathrm{MT} \left( R_{\mathrm{MT}_0}, U_{\mathrm{MT}_0} \right) = p_\mathrm{max}$ we say that the transmission is truncated by $p_\mathrm{max}$ whereas we say that it is truncated by $i_0$ if $p_\mathrm{MT} \left( R_{\mathrm{MT}_0}, U_{\mathrm{MT}_0} \right) = i_0  \left( \tau U_{\mathrm{MT}_0} \right)^\alpha$.

We define the event $\mathcal{Q}^{(m)}_{\mathrm{MT}_i}$ as: \emph{the most interfered BS by MT$_i$'s transmissions belong to tier m}. Hence we can define the event $\mathcal{X}^{(j,m)}_{\mathrm{MT}_i} = \mathcal{X}^{(j)}_{\mathrm{MT}_i} \cap \mathcal{Q}^{(m)}_{\mathrm{MT}_i}$ which means: \emph{MT$_i$ is associated with tier j and the most interfered BS by MT$_i$'s transmission belong to tier m}. Mathematically this event can be expressed as
\begin{align}
\label{eq:Xjm}
	&\mathcal{X}^{(j,m)}_{\mathrm{MT}_i} = \mathcal{X}^{(j)}_{\mathrm{MT}_i} \cap 
	\overbrace{ \left\{ R^{(j)}_{\mathrm{MT}_i,(2)} > R^{(m)}_{\mathrm{MT}_i,(1)} \right\}}
	^{\mathcal{Q}^{(m)}_{\mathrm{MT}_i}},\, \mathrm{if} j \neq m \nonumber \\	
	&\mathcal{X}^{(j,j)}_{\mathrm{MT}_i} = \mathcal{X}^{(j)}_{\mathrm{MT}_i} \cap 
	\overbrace{ \left\{ R^{(j)}_{\mathrm{MT}_i,(2)} < R^{(\tilde{j})}_{\mathrm{MT}_i,(1)}
	\right\} }
	^{\mathcal{Q}^{(j)}_{\mathrm{MT}_i}} ,\, \mathrm{if} j = m 
\end{align}

\begin{figure}[t]
\centering
\includegraphics[width=3in]{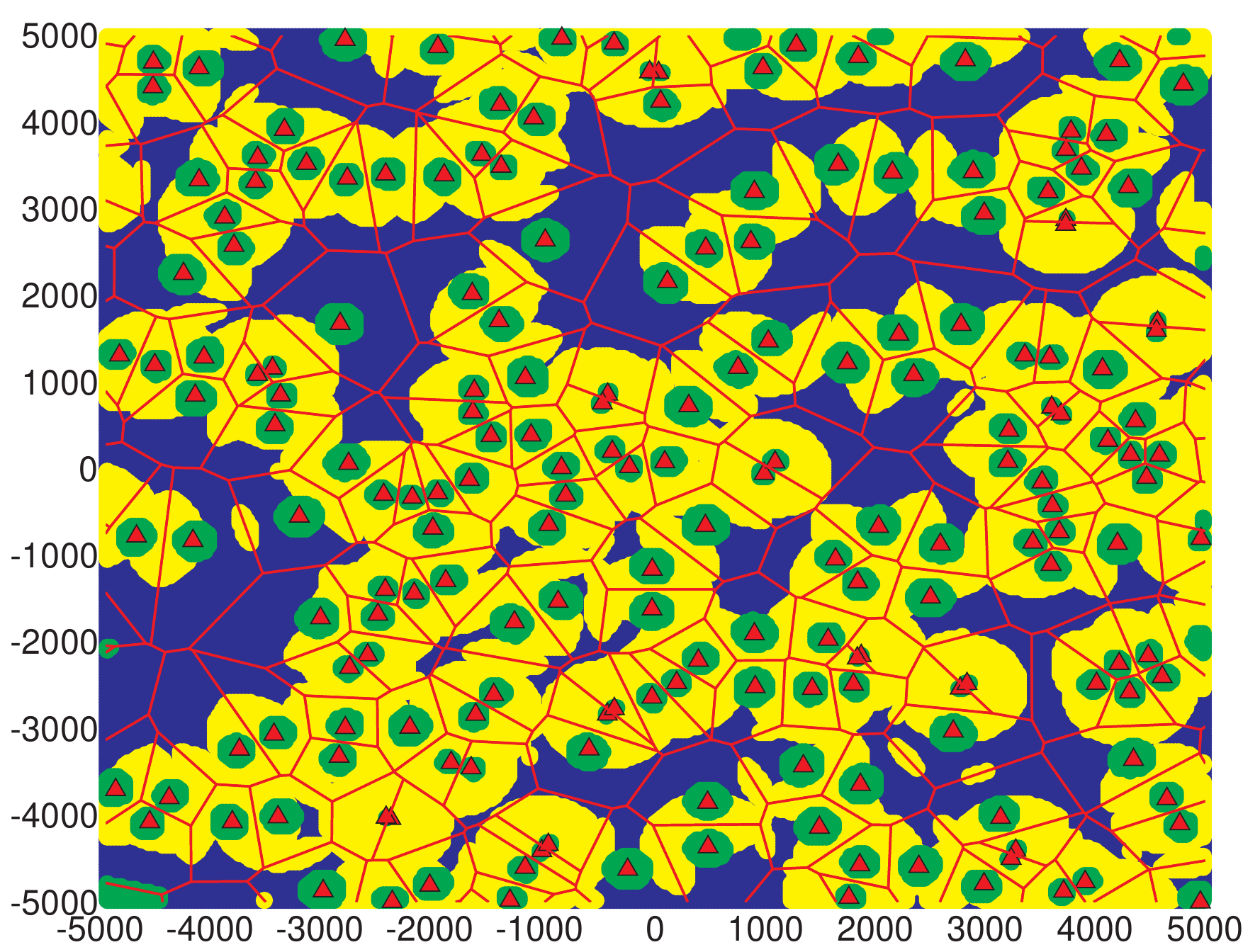}
\caption{Homogeneous network realization in $[-5000, 5000]^2$ m$^2$ showing with green color positions related to non truncated MTs. Yellow and blue colors are associated with MTs truncated by $i_0$ and $p_\mathrm{max}$ respectively. The simulation parameters are $\lambda^{(1)}=2$ BS/km$^2$, $i_0=-90$ dBm, $p_\mathrm{max} = 30$ dBm, $p_0=-70$ dBm and $\epsilon=1$.}
\label{fig:Fig1}
\end{figure}

It can be observed from Fig. \ref{fig:Fig1} that the locations of non truncated MTs are different for each BS. This is due to the fact that such locations also depend on the positions of neighboring BSs. It can be noticed that MTs truncated by $i_0$ tend to be placed in positions close to a victim BS whereas MTs truncated by $p_\mathrm{max}$ tend to be placed far from their serving BSs but also far from any other BSs. Intuitively this mechanism allows to reach a balance between interference and desired received power that increases the SINR in statistical terms. It is important to note that MTs that tend to cause more interference than $i_0$ or try to transmit with more power than $p_\mathrm{max}$ are not keep silent; instead they transmit at maximum power that do not violate the aforementioned conditions. 

\subsection{SINR}
\label{sec:SINR}
The SINR for the typical MT$_0$ can be expressed as
\begin{equation}
\label{eq:SINR}
	\mathrm{SINR}_{\mathrm{MT}_0} = \frac{ H_{\mathrm{MT}_0} \left(\tau R_{\mathrm{MT}_0} \right)^{-\alpha}
	P_{\mathrm{MT}_0} }{I_\mathrm{exact} + \sigma_n^2 }
\end{equation}
being $H_{\mathrm{MT}_0}$ the multi-path fading for the desired link, $R_{\mathrm{MT}_0}$ the distance to the serving BS including shadowing, $P_{\mathrm{MT}_0}$ the transmit power, $I_\mathrm{exact}$ the aggregate interference and $\sigma_n^2$ the noise power. 
As it has been mentioned in the introduction, in the UL the interfering MT does not follow a PPP even when the positions of BSs and MTs are PPPs. The nature of such PP which depends on the scheduling and association criteria makes the analysis intractable, so the PPP assumption of interfering MT locations seems to be appealing.  
However, it is necessary to add appropriate correlation between the probe BS and the interfering MTs' locations to improve the accuracy of the analysis \cite{ElSawy14, Singh15}. Note that in the UL an interfering MT can be placed closer to the probe BS than the probe MT. 
Nevertheless, the interfering MT always has higher weighted received power from its serving BS than from the probe BS thanks to the association criteria. 
We use this fact, which has been previously used in \cite{ElSawy14} and \cite{Singh15}, to model the interfering MTs' locations. To do that we perform a conditional thinning over the event $\mathcal{O}^{(j,k)}_{\mathrm{MT}_i}$ defined as: \emph{the interfering MT$_i$ belonging to tier $k$ receives higher weighted power from its serving BS than from the probe BS which belongs to tier $j$}. Mathematically we have
\begin{equation}
\label{eq:O_MTi}
	\mathcal{O}^{(j,k)}_{\mathrm{MT}_i} = \left\{ t^{(k)} 
	\left( \tau R_{\mathrm{MT}_i} \right)^{-\alpha} > 
	t^{(j)} \left( \tau D_{\mathrm{MT}_i} \right)^{-\alpha} \right\}
\end{equation}
being $R_{\mathrm{MT}_i}$ the distance between the interfering MT$_i$ and its serving BS whereas $D_{\mathrm{MT}_i}$ is the distance between MT$_i$ and the probe BS. 

Additionally, interference awareness involves a different kind of correlation between MT$_i$ and the probe BS position since MT$_i$'s transmission cannot cause higher interference level than $i_0$. This correlation is also added by means of a dependent thinning with the event $\mathcal{Z}_{\mathrm{MT}_i}$ which is defined as: \emph{the interfering MT$_i$ causes less interference to $\mathrm{BS}_0$ than $i_0$}. More formally we have
\begin{equation}
\label{eq:ZMTi}
\mathcal{Z}_{\mathrm{MT}_i} = 
    \left\{ P_{\mathrm{MT}_i} \left( \tau D_{\mathrm{MT}_i} \right)^{-\alpha} 
     < i_0 \right\}    
\end{equation}
%
Finally, the exact interference term $I_\mathrm{exact}$, which is intractable, is approximated as $I_\mathrm{exact} \simeq I$. The analytical interference term $I$ for a probe MT associated to tier $j$ appears below
\begin{equation}
\label{eq:I}
	I = \sum_{k \in \mathcal{K}} \sum_{\mathrm{MT}_i \in \Psi^{(k)}} 
	H_{\mathrm{MT}_i}  \left(\tau D_{\mathrm{MT}_i} \right)^{-\alpha} P_{\mathrm{MT}_i} 
	\mathbf{1} \left( \mathcal{O}^{(j,k)}_{\mathrm{MT}_i} \right)
	\mathbf{1} \left( \mathcal{Z}_{\mathrm{MT}_i} \right)
\end{equation}
where $\Psi^{(k)}$ represents a PPP with the locations of interfering MTs scheduled in the RB of interest, so its density is $\lambda^{(k)}$.
From (\ref{eq:I}) it can be observed that the interfering MTs' locations of tier $k$ are obtained performing a conditional thinning over $\Psi^{(k)}$ using the event $\mathcal{O}^{(j,k)}_{\mathrm{MT}_i}$ that discards locations that would result associated with the probe BS and the event $\mathcal{Z}_{\mathrm{MT}_i}$ that discards locations that would cause higher interference than $i_0$. 

\section{Analysis of Interference Aware Power Control}
\label{sec:Analysis of Interference Aware Power Control}
In this section the average transmit power, mean and variance of the interference, the ccdf of the SINR and the average SE are obtained. 
The probability of being associated to tier $j$ has been previously obtained in \cite{Singh15} and is reproduced here for convenience
$
\Pr \left( \mathcal{X}_{\mathrm{MT}_{0}}^{(j)} \right)={\lambda ^{(j)}}/{\sum\limits_{k\in \mathcal{K}}{\left( \frac{t^{(k)}}{t^{(j)}} \right)^{\frac{2}{\alpha }}\lambda ^{(k)}}}
$
. 

The probability of being associated to tier $j$ and being the most interfered BS from tier $m$ is given with the following proposition. 
\begin{lemma}
\label{prop:Soft IAFPC PrXj}
The probability of the event $\mathcal{X}_{\mathrm{MT}_{0}}^{(j,m)}$ with $j \neq m$ for the typical MT using IAFPC is
\ifOneColumn
\begin{align}
\Pr & \left( \mathcal{X}_{\mathrm{MT}_{0}}^{(j,m)} \right)=\frac{\lambda ^{(j)}\lambda ^{(m)}\left( \frac{t^{(j)}}{t^{(m)}} \right)^{\frac{2}{\alpha }}}{\left( \lambda ^{(j)}+\lambda ^{(m)} \right)^{2}}\mathbf{1}\left( \left( \frac{t^{(j)}}{t^{(m)}} \right)>1 \right)+ 
\nonumber \\ 
& \left( \frac{\lambda ^{(m)}\left( 2\lambda ^{(j)}+\lambda ^{(m)} \right)}{\left( \lambda ^{(j)}+\lambda ^{(m)} \right)^{2}}-\frac{\lambda ^{(m)}}{\lambda ^{(j)}\left( \frac{t^{(j)}}{t^{(m)}} \right)^{\frac{2}{\alpha }}+\lambda ^{(m)}} \right) 
 \times \mathbf{1}\left( \left( \frac{t^{(j)}}{t^{(m)}} \right)\le 1 \right) 
\end{align}
\else
\begin{align}
\Pr & \left( \mathcal{X}_{\mathrm{MT}_{0}}^{(j,m)} \right)=\frac{\lambda ^{(j)}\lambda ^{(m)}\left( \frac{t^{(j)}}{t^{(m)}} \right)^{\frac{2}{\alpha }}}{\left( \lambda ^{(j)}+\lambda ^{(m)} \right)^{2}}\mathbf{1}\left( \left( \frac{t^{(j)}}{t^{(m)}} \right)>1 \right)+ 
\nonumber \\ 
& \left( \frac{\lambda ^{(m)}\left( 2\lambda ^{(j)}+\lambda ^{(m)} \right)}{\left( \lambda ^{(j)}+\lambda ^{(m)} \right)^{2}}-\frac{\lambda ^{(m)}}{\lambda ^{(j)}\left( \frac{t^{(j)}}{t^{(m)}} \right)^{\frac{2}{\alpha }}+\lambda ^{(m)}} \right) 
\nonumber \\
& \times \mathbf{1}\left( \left( \frac{t^{(j)}}{t^{(m)}} \right)\le 1 \right) 
\end{align}
\fi

If $j=m$ the probability is given below

\ifOneColumn
\begin{align}
\Pr & \left( \mathcal{X}_{\mathrm{MT}_{0}}^{(j,j)} \right)=\frac{\left( \lambda ^{(j)} \right)^{2}}{\left( \lambda ^{(j)}+\lambda ^{(\tilde{j})} \right)^{2}}\mathbf{1}\left( \left( \frac{t^{(\tilde{j})}}{t^{(j)}} \right)\le 1 \right)+ 
\nonumber \\ 
& \frac{\left( \lambda ^{(j)} \right)^{2}\left( \frac{t^{(j)}}{t^{(\tilde{j})}} \right)^{\frac{2}{\alpha }}\left( \lambda ^{(j)}-\lambda ^{(\tilde{j})}\left( \left( \frac{t^{(j)}}{t^{(\tilde{j})}} \right)^{\frac{2}{\alpha }}-2 \right) \right)}{\left( \lambda ^{(j)}+\lambda ^{(\tilde{j})} \right)^{2}\left( \lambda ^{(\tilde{j})}+\lambda ^{(j)}\left( \frac{t^{(j)}}{t^{(\tilde{j})}} \right)^{\frac{2}{\alpha }} \right)}
\times \mathbf{1}\left( \left( \frac{t^{(\tilde{j})}}{t^{(j)}} \right)>1 \right) 
\end{align}
\else
\begin{align}
\Pr & \left( \mathcal{X}_{\mathrm{MT}_{0}}^{(j,j)} \right)=\frac{\left( \lambda ^{(j)} \right)^{2}}{\left( \lambda ^{(j)}+\lambda ^{(\tilde{j})} \right)^{2}}\mathbf{1}\left( \left( \frac{t^{(\tilde{j})}}{t^{(j)}} \right)\le 1 \right)+ 
\nonumber \\ 
& \frac{\left( \lambda ^{(j)} \right)^{2}\left( \frac{t^{(j)}}{t^{(\tilde{j})}} \right)^{\frac{2}{\alpha }}\left( \lambda ^{(j)}-\lambda ^{(\tilde{j})}\left( \left( \frac{t^{(j)}}{t^{(\tilde{j})}} \right)^{\frac{2}{\alpha }}-2 \right) \right)}{\left( \lambda ^{(j)}+\lambda ^{(\tilde{j})} \right)^{2}\left( \lambda ^{(\tilde{j})}+\lambda ^{(j)}\left( \frac{t^{(j)}}{t^{(\tilde{j})}} \right)^{\frac{2}{\alpha }} \right)}
\nonumber \\
&\times \mathbf{1}\left( \left( \frac{t^{(\tilde{j})}}{t^{(j)}} \right)>1 \right) 
\end{align}
\fi

\end{lemma}
\begin{proof}
See Appendix \ref{app:Soft IAFPC PrXj}.
\end{proof}
As it can be observed from section \ref{sec:Power Control Mechanism} the computation of the transmit power (\ref{eq:pMT Soft})  requires knowledge of the joint pdf of distances towards the serving and most interfering BSs. This fact complicates further the analysis since it requires a joint pdf of distances and to deal with a nonlinear function, i.e. $\mathrm{min}(\cdot)$ function. The next proposition give such joint pdf of distances.
\begin{lemma}
\label{prop:Soft IAFPC f_RMT_UMT Cond Xjm}
The joint pdf of distances towards the serving and most interfered BS conditioned on event $\mathcal{X}_{\mathrm{MT}_0}^{(j,m)}$ for the typical MT using IAFPC is given as

\ifOneColumn
\begin{align}
\label{eq:Soft IAFPC Soft IAFPC f_RMT_UMT Cond Xjm j neq m}
& f_{R_{\mathrm{MT}_{0}},U_{\mathrm{MT}_{0}}}  \left( v,w|\mathcal{X}_{\mathrm{MT}_{0}}^{(j,m)} \right) =
  \frac{f_{R_{\mathrm{MT}_{0},(1)}^{(m)}}\left( w \right)\zeta ^{(j)}\left( v,w \right)}{\Pr 
  \left( \mathcal{X}_{\mathrm{MT}_{0}}^{(j,m)} \right)} 
  \mathbf{1} \left( w>\left( \frac{t^{(m)}}{t^{(j)}} \right)^\frac{1}{\alpha} v \right)
\end{align}
\else
\begin{align}
\label{eq:Soft IAFPC Soft IAFPC f_RMT_UMT Cond Xjm j neq m}
& f_{R_{\mathrm{MT}_{0}},U_{\mathrm{MT}_{0}}}  \left( v,w|\mathcal{X}_{\mathrm{MT}_{0}}^{(j,m)} \right) =
\nonumber \\
&\quad
  \frac{f_{R_{\mathrm{MT}_{0},(1)}^{(m)}}\left( w \right)\zeta ^{(j)}\left( v,w \right)}{\Pr 
  \left( \mathcal{X}_{\mathrm{MT}_{0}}^{(j,m)} \right)} 
  \mathbf{1} \left( w>\left( \frac{t^{(m)}}{t^{(j)}} \right)^\frac{1}{\alpha} v \right)
\end{align}
\fi

if $j \neq m$ and

\ifOneColumn
\begin{align}
\label{eq:Soft IAFPC Soft IAFPC f_RMT_UMT Cond Xjm j = m}
& f_{R_{\mathrm{MT}_{0}},U_{\mathrm{MT}_{0}}}  \left( v,w|\mathcal{X}_{\mathrm{MT}_{0}}^{(j,m)} \right) = 
    \frac{
    \bar{F}_{R_{\mathrm{MT}_{0},(1)}^{(\tilde{j})}}\left( \mathrm{max} \left( \left( \frac{t^{(\tilde{j})}}{t^{(j)}}
    \right)^{\frac{1}{\alpha }}v,w \right) \right)}{\Pr \left( \mathcal{X}_{\mathrm{MT}_{0}}^{(j,j)} 
    \right)}
     f_{R_{\mathrm{MT}_{0},(1)}^{(j)},R_{\mathrm{MT}_{0},(2)}^{(j)}}\left( v,w \right) 
    \mathbf{1} \left( w > v \right)
\end{align}
\else
\begin{align}
\label{eq:Soft IAFPC Soft IAFPC f_RMT_UMT Cond Xjm j = m}
& f_{R_{\mathrm{MT}_{0}},U_{\mathrm{MT}_{0}}}  \left( v,w|\mathcal{X}_{\mathrm{MT}_{0}}^{(j,m)} \right) = \nonumber \\
& \quad     
    \frac{
    \bar{F}_{R_{\mathrm{MT}_{0},(1)}^{(\tilde{j})}}\left( \mathrm{max} \left( \left( \frac{t^{(\tilde{j})}}{t^{(j)}}
    \right)^{\frac{1}{\alpha }}v,w \right) \right)}{\Pr \left( \mathcal{X}_{\mathrm{MT}_{0}}^{(j,j)} 
    \right)}
    \nonumber \\
    & \quad \times f_{R_{\mathrm{MT}_{0},(1)}^{(j)},R_{\mathrm{MT}_{0},(2)}^{(j)}}\left( v,w \right) 
    \mathbf{1} \left( w > v \right)
\end{align}
\fi
\noindent with $j=m$  where
\begin{align}
\label{eq:Soft IAFPC zeta(v,w)}
\zeta ^{(j)}\left( v,w \right)=2\pi \lambda ^{(j)}v\mathrm{e}^{-\pi \lambda ^{(j)}\mathrm{max} ^{2}\left( v,w \right)}
\end{align}
The joint pdf of the nearest and second nearest point has been obtained in \cite{Martin-Vega14} and is given by 
\begin{equation}
\label{eq:joint pdf R1, R2}
f_{R^{(j)}_{\mathrm{MT}_0,(1)}, R^{(j)}_{\mathrm{MT}_0,(2)}}(r_1,r_2) =
 4 \left( \pi \lambda^{(j)} \right)^2 r_1 r_2  \mathrm{e}^{-\pi \lambda^{(j)} r_2^2}, \, r_1 < r_2
\end{equation}
\end{lemma}
\begin{proof}
See Appendix \ref{app:Soft IAFPC f_RMT_UMT Cond Xjm}.
\end{proof}
The joint pdf of distances given in the previous proposition allows to obtain the average transmitted power as

\ifOneColumn
\begin{align}
\mathbb{E}\left[ P_{\mathrm{MT}_{0}} \right] = 
\sum\limits_{j\in \mathcal{K}}{\sum\limits_{m\in \mathcal{K}}{\Pr \left( \mathcal{X}_{\mathrm{MT}_{0}}^{(j,m)} \right)\times }} 
 \mathbb{E}_{R_{\mathrm{MT}_{0}},U_{\mathrm{MT}_{0}}}\left[ p_{\mathrm{MT}}\left( R_{\mathrm{MT}_{0}},U_{\mathrm{MT}_{0}} \right)|\mathcal{X}_{\mathrm{MT}_{0}}^{(j,m)} \right] 
\end{align}
\else
\begin{align}
\mathbb{E}\left[ P_{\mathrm{MT}_{0}} \right] &= 
\sum\limits_{j\in \mathcal{K}}{\sum\limits_{m\in \mathcal{K}}{\Pr \left( \mathcal{X}_{\mathrm{MT}_{0}}^{(j,m)} \right)\times }} 
\nonumber \\ 
& \mathbb{E}_{R_{\mathrm{MT}_{0}},U_{\mathrm{MT}_{0}}}\left[ p_{\mathrm{MT}}\left( R_{\mathrm{MT}_{0}},U_{\mathrm{MT}_{0}} \right)|\mathcal{X}_{\mathrm{MT}_{0}}^{(j,m)} \right] 
\end{align}
\fi
\noindent where the function of the transmitted power according to the distance to the serving and most interfered BS is given in (\ref{eq:pMT Soft}). The Laplace transform of the interference is given with the following proposition
\begin{lemma}
\label{prop:Soft IAFPC LI}
The Laplace transform of the interference for the typical MT using IAFPC is
\begin{equation}
\label{eq:Soft IAFPC LI}
\mathcal{L}_{I}\left( s|\mathcal{X}_{\mathrm{MT}_{0}}^{(j)} \right)=\exp \left( \beta ^{(j)}\left( s \right) \right)
\end{equation}
where
%
\begin{align}
\label{eq:Soft IAFPC theta^j}
& \beta ^{(j)}\left( s \right)
= - \sum\limits_{k\in \mathcal{K}}{2\pi \lambda ^{(k)}\sum\limits_{n\in \mathcal{K}}{\Pr \left( \mathcal{Q}_{\mathrm{MT}_{i}}^{(n)}|\mathcal{X}_{\mathrm{MT}_{i}}^{(k)} \right)}}  
\nonumber \\ 
& \int\limits_{r=0}^{\infty }{\int\limits_{u=\left( \frac{t^{(n)}}{t^{(k)}} \right)^{\frac{1}{\alpha }}r}^{\infty }{f_{R_{\mathrm{MT}_{i}},U_{\mathrm{MT}_{i}}}\left( r,u|\mathcal{X}_{\mathrm{MT}_{i}}^{(k,n)} \right)}} \chi \left( s,r,u \right)\mathrm{d}r\mathrm{d}u 
\end{align}
\noindent and
\ifOneColumn
\begin{align}
& \chi \left( s,r,u \right)=\frac{sp_{\mathrm{MT}}\left( r,u \right)\tau ^{-\alpha }}{\alpha -2} 
\mathrm{max}^{2-\alpha }\left( \left( \frac{t^{(j)}}{t^{(k)}} \right)^{\frac{1}{\alpha }}r,\frac{1}{\tau }\left( \frac{p_{\mathrm{MT}}\left( r,u \right)}{i_{0}} \right)^{\frac{1}{\alpha }} \right) 
\nonumber \\ 
& \quad \,_{2}F_{1} \Bigg(1,\frac{\alpha -2}{\alpha };2-\frac{2}{\alpha };-sp_{\mathrm{MT}}\left( r,u \right)\tau ^{-\alpha } 
\mathrm{max}^{-\alpha }\left( \left( \frac{t^{(j)}}{t^{(k)}} \right)^{\frac{1}{\alpha }}r,\frac{1}{\tau }\left( \frac{p_{\mathrm{MT}}\left( r,u \right)}{i_{0}} \right)^{\frac{1}{\alpha }} \right)\Bigg) 
\end{align}
\else
\begin{align}
& \chi \left( s,r,u \right)=\frac{sp_{\mathrm{MT}}\left( r,u \right)\tau ^{-\alpha }}{\alpha -2} 
\nonumber \\ 
& \quad \mathrm{max}^{2-\alpha }\left( \left( \frac{t^{(j)}}{t^{(k)}} \right)^{\frac{1}{\alpha }}r,\frac{1}{\tau }\left( \frac{p_{\mathrm{MT}}\left( r,u \right)}{i_{0}} \right)^{\frac{1}{\alpha }} \right) 
\nonumber \\ 
& \quad \,_{2}F_{1} \Bigg(1,\frac{\alpha -2}{\alpha };2-\frac{2}{\alpha };-sp_{\mathrm{MT}}\left( r,u \right)\tau ^{-\alpha } 
\nonumber \\ 
& \quad \mathrm{max}^{-\alpha }\left( \left( \frac{t^{(j)}}{t^{(k)}} \right)^{\frac{1}{\alpha }}r,\frac{1}{\tau }\left( \frac{p_{\mathrm{MT}}\left( r,u \right)}{i_{0}} \right)^{\frac{1}{\alpha }} \right)\Bigg) 
\end{align}
\fi
\noindent being $p_{\mathrm{MT}}\left( r,u \right)$ given by (\ref{eq:pMT Soft}).
\end{lemma}
\begin{proof}
See Appendix \ref{app:Soft IAFPC LI}.
\end{proof}

The Laplace transform of the interference allows to obtain the mean and the variance of the interference using the first and second derivatives of $\beta^{(j)}(s)$. Those metrics are given with the next proposition.

\begin{proposition}
\label{prop:Soft IAFPC E(I) var(I)}
The mean and variance of the interference of the typical MT are given as follows
\ifOneColumn
\begin{align}
\label{eq:Soft IAFPC var(I)}
    \mathbb{E}\left[ I \right] &= -\sum\limits_{j\in \mathcal{K}}
    {\Pr \left(\mathcal{X}_{\mathrm{MT}_{0}}^{(j)} \right)\beta '^{(j)}\left( 0 \right)} 
    \\
    \operatorname{var}\left( I \right) &=-\sum\limits_{j\in \mathcal{K}}{\Pr \left(
    \mathcal{X}_{\mathrm{MT}_{0}}^{(j)} \right)} 
     \left( \beta ''^{(j)}\left( 0 \right)+\left( \beta '^{(j)}\left( 0 \right)
    \right)^{2}-\left( \mathbb{E}\left[ I \right] \right)^{2} \right) 
\end{align}
\else
\begin{align}
\label{eq:Soft IAFPC var(I)}
    \mathbb{E}\left[ I \right] &= -\sum\limits_{j\in \mathcal{K}}
    {\Pr \left(\mathcal{X}_{\mathrm{MT}_{0}}^{(j)} \right)\beta '^{(j)}\left( 0 \right)} 
    \\
    \operatorname{var}\left( I \right) &=-\sum\limits_{j\in \mathcal{K}}{\Pr \left(
    \mathcal{X}_{\mathrm{MT}_{0}}^{(j)} \right)} 
    \nonumber \\ 
    & \left( \beta ''^{(j)}\left( 0 \right)+\left( \beta '^{(j)}\left( 0 \right)
    \right)^{2}-\left( \mathbb{E}\left[ I \right] \right)^{2} \right) 
\end{align}
\fi
\noindent being 
\ifOneColumn
\begin{align}
& \beta '^{(j)}\left( 0 \right)=-\sum\limits_{k\in \mathcal{K}}{2\pi \lambda ^{(k)}}\sum\limits_{n\in \mathcal{K}}{\Pr \left( \mathcal{Q}_{\mathrm{MT}_{i}}^{(n)}|\mathcal{X}_{\mathrm{MT}_{i}}^{(k)} \right) } 
 \int\limits_{r=0}^{\infty }{\int\limits_{u=\left( \frac{t^{(n)}}{t^{(k)}} \right)^{\frac{1}{\alpha }}r}^{\infty }{f_{R_{\mathrm{MT}_{i}},U_{\mathrm{MT}_{i}}}\left( r,u|\mathcal{X}_{\mathrm{MT}_{i}}^{(k,n)} \right)}}
\nonumber \\ 
& \frac{\tau ^{-\alpha }p_{\mathrm{MT}}\left( r,u \right)}{\alpha -2}  \mathrm{max}^{2-\alpha }\left( \left( \frac{t^{(j)}}{t^{(k)}} \right)^{\frac{1}{\alpha }}r,
\frac{1}{\tau }\left( \frac{p_{\mathrm{MT}}\left( r,u \right)}{i_{0}} \right)^{\frac{1}{\alpha }}
\right)\mathrm{d}u\mathrm{d}r  
\end{align}
\else
\begin{align}
& \beta '^{(j)}\left( 0 \right)=-\sum\limits_{k\in \mathcal{K}}{2\pi \lambda ^{(k)}}\sum\limits_{n\in \mathcal{K}}{\Pr \left( \mathcal{Q}_{\mathrm{MT}_{i}}^{(n)}|\mathcal{X}_{\mathrm{MT}_{i}}^{(k)} \right)\times } 
\nonumber \\ 
& \; \int\limits_{r=0}^{\infty }{\int\limits_{u=\left( \frac{t^{(n)}}{t^{(k)}} \right)^{\frac{1}{\alpha }}r}^{\infty }{f_{R_{\mathrm{MT}_{i}},U_{\mathrm{MT}_{i}}}\left( r,u|\mathcal{X}_{\mathrm{MT}_{i}}^{(k,n)} \right)}}\frac{\tau ^{-\alpha }p_{\mathrm{MT}}\left( r,u \right)}{\alpha -2} 
\nonumber \\ 
& \; \mathrm{max}^{2-\alpha }\left( \left( \frac{t^{(j)}}{t^{(k)}} \right)^{\frac{1}{\alpha }}r,
\frac{1}{\tau }\left( \frac{p_{\mathrm{MT}}\left( r,u \right)}{i_{0}} \right)^{\frac{1}{\alpha }}
\right)\mathrm{d}u\mathrm{d}r  
\end{align}
\fi
\ifOneColumn
\begin{align}
& \beta ''^{(j)}\left( 0 \right)= - \sum\limits_{k\in \mathcal{K}}{2\pi \lambda ^{(k)}}\sum\limits_{n\in \mathcal{K}}{\Pr \left( \mathcal{Q}_{\mathrm{MT}_{i}}^{(n)}|\mathcal{X}_{\mathrm{MT}_{i}}^{(k)} \right)\times } 
\int\limits_{r=0}^{\infty }{\int\limits_{u=\left( \frac{t^{(n)}}{t^{(k)}} \right)^{\frac{1}{\alpha }}r}^{\infty }{f_{R_{\mathrm{MT}_{i}},U_{\mathrm{MT}_{i}}}\left( r,u|\mathcal{X}_{\mathrm{MT}_{i}}^{(k,n)} \right)}}
\nonumber \\ 
& \quad \frac{\left( \tau ^{-\alpha }p_{\mathrm{MT}}\left( r,u \right) \right)^{2}}{1-\alpha}   \mathrm{max}^{2\left( 1-\alpha  \right)}\left( \left( \frac{t^{(j)}}{t^{(k)}} \right)^{\frac{1}{\alpha }}r,
\frac{1}{\tau }\left( \frac{p_{\mathrm{MT}}\left( r,u \right)}{i_{0}} \right)^{\frac{1}{\alpha }}
\right)\mathrm{d}u\mathrm{d}r 
\end{align}
\else
\begin{align}
& \beta ''^{(j)}\left( 0 \right)= - \sum\limits_{k\in \mathcal{K}}{2\pi \lambda ^{(k)}}\sum\limits_{n\in \mathcal{K}}{\Pr \left( \mathcal{Q}_{\mathrm{MT}_{i}}^{(n)}|\mathcal{X}_{\mathrm{MT}_{i}}^{(k)} \right)\times } 
\nonumber \\ 
& \; \int\limits_{r=0}^{\infty }{\int\limits_{u=\left( \frac{t^{(n)}}{t^{(k)}} \right)^{\frac{1}{\alpha }}r}^{\infty }{f_{R_{\mathrm{MT}_{i}},U_{\mathrm{MT}_{i}}}\left( r,u|\mathcal{X}_{\mathrm{MT}_{i}}^{(k,n)} \right)}}\frac{\left( \tau ^{-\alpha }p_{\mathrm{MT}}\left( r,u \right) \right)^{2}}{1-\alpha} 
\nonumber \\ 
& \;  \mathrm{max}^{2\left( 1-\alpha  \right)}\left( \left( \frac{t^{(j)}}{t^{(k)}} \right)^{\frac{1}{\alpha }}r,
\frac{1}{\tau }\left( \frac{p_{\mathrm{MT}}\left( r,u \right)}{i_{0}} \right)^{\frac{1}{\alpha }}
\right)\mathrm{d}u\mathrm{d}r 
\end{align}
\fi

\end{proposition}

\begin{proof}
The proof consists on expressing the mean and variance conditioned on $\mathcal{X}_{\mathrm{MT}_{0}}^{(j)}$ and then obtaining the first and second derivatives of the Laplace transform of the interference evaluated at $s=0$ to obtain its moments. 
\end{proof}

Finally, the ccdf of the SINR can be obtained as appears in (\ref{eq:Soft IAFPC ccdf SINR}), 
\begin{figure*}[t]
\normalsize 
\begin{align}
\label{eq:Soft IAFPC ccdf SINR}
& \bar{F}_{\mathrm{SINR}}\left( \gamma  \right)=\sum\limits_{j\in \mathcal{K}}{\sum\limits_{m\in \mathcal{K}}{\Pr \left( \mathcal{X}_{\mathrm{MT}_{0}}^{(j,m)} \right)}}\int\limits_{v=0}^{\infty }{\int\limits_{w=\left( \frac{t^{(m)}}{t^{(j)}} \right)^{\frac{1}{\alpha }}v}^{\infty }{f_{R_{\mathrm{MT}_{0}},U_{\mathrm{MT}_{0}}}\left( v,w|\mathcal{X}_{\mathrm{MT}_{0}}^{(j,m)} \right)}}\mathrm{e}^{-\frac{\gamma \sigma _{n}^{2}\left( \tau v \right)^{\alpha }}{p_{\mathrm{MT}}\left( v,w \right)}}
\nonumber \\
& \mathcal{L}_{I}\left( \frac{\gamma \left( \tau v \right)^{\alpha }}{p_{\mathrm{MT}}\left( v,w \right)}|\mathcal{X}_{\mathrm{MT}_{0}}^{(j)} \right)\mathrm{d}v\mathrm{d}w
\end{align}
\hrulefill
\end{figure*}
where it has been applied the total probability theorem and performed expectation over the fading. 

The SE of the typical MT is expressed in bits per second per Hertz (bps/Hz) and represents how well the spectrum of the transmission of a randomly selected MT is exploited. Hence this metric is directly related to its SINR, which is given below using the well known Shannon formula $\mathrm{SE}_{\mathrm{MT}_{0}}=\log _{2}\left( 1+\mathrm{SINR}_{\mathrm{MT}_{0}} \right)$.
The ccdf of the SE of the typical MT can be expressed as 
\begin{align}
\label{eq:Shannon ccdf SE}
 \bar{F}_{\mathrm{SE}}\left( \xi  \right)= \sum\limits_{j\in \mathcal{K}}{\sum\limits_{m\in \mathcal{K}}{\Pr \left( \mathcal{X}_{\mathrm{MT}_{0}}^{(j,m)} \right)}} 
\bar{F}_{\mathrm{SINR}}\left( 2^{\xi }-1|\mathcal{X}_{\mathrm{MT}_{0}}^{(j,m)} \right) 
\end{align}
\noindent where we have applied the total probability theorem and used the inequality  
$ \mathrm{SINR}_{\mathrm{MT}_0} > 2^{\xi }-1$. From (\ref{eq:Shannon ccdf SE}) it can be obtained the average SE, 
$ \mathbb{E}_{\mathrm{SE}}\left[ \mathrm{SE} \right] $, using the fact that if X is a positive RV then 
$\mathbb{E}[X]= \int_{x>0} f_X(t) \mathrm{d}t$. 

It should be noticed that the ccdf of the SINR in (\ref{eq:Soft IAFPC ccdf SINR}) has four nested integrals since the Laplace transform of the interference given in (\ref{eq:Soft IAFPC theta^j}) has two nested integrals for IAFPC. Hence approximations of the Laplace transform of the interference in the form $\mathcal{L}_{I}\left( s|\mathcal{X}_{\mathrm{MT}_{0}}^{(j)} \right)\simeq \mathcal{L}_{{\hat{I}}}\left( s|\mathcal{X}_{\mathrm{MT}_{0}}^{(j)} \right)$ are proposed in next section where $\mathcal{L}_{{\hat{I}}}\left( s|\mathcal{X}_{\mathrm{MT}_{0}}^{(j)} \right)$ has a closed form expression. Additionally asymptotic analysis is performed in Section \ref{sec:Asymptotic Analysis}. This latter approach avoids the need of approximating the interference term since expressions are further simplified.

%
\section{Statistical Modeling of the interference}
\label{sec:Statistical Modeling of the interference}
In this section two approaches are proposed: (i) approximate the Laplace transform through a sigmoidal logistic function whose parameters are obtained by means of logistic regression and (ii) approximate the Laplace transform by a suitable function and then perform MM in order to obtain the function parameters. Those approaches will be used to approximate the Laplace transform of IAFPC in order to reduce the computational complexity. 

%
%
\subsection{Sigmoidal Approximation}
\label{sec:Sigmoidal approximation}
The Laplace transform of the interference given in (\ref{eq:Soft IAFPC theta^j}) has a S-shape if we represent the s-axis in dBs. Hence an approximation with a sigmoidal logistic function \cite{Costarelli14} is proposed as follows
\begin{equation}
\label{eq:Soft IAFPC logistic function 1}
\mathcal{L}_{I}\left( s^{(\mathrm{dB})}|\mathcal{X}_{\mathrm{MT}_{0}}^{(j)} \right)\simeq g_{I}\left( s^{(\mathrm{dB})} \right)=\frac{1}{1+e^{b_{0}\left( s^{(\mathrm{dB})}-s^{(\mathrm{dB})}_{0} \right)}}
\end{equation}
where $s^{(\mathrm{dB})}=10 \log_{10} (s)$ and $g_I (s^{(\mathrm{dB})})$ is a sigmoidal function with two parameters $b_0$ and $s^{(\mathrm{dB})}_0$ with the following properties
\ifOneColumn
\begin{align}
& \underset{s^{(\mathrm{dB})}\to -\infty }{\mathop{\lim }}\,g_{I}\left( s^{(\mathrm{dB})} \right)=1;\underset{s^{(\mathrm{dB})}\to +\infty }{\mathop{\lim }}\,g_{I}\left( s^{(\mathrm{dB})} \right)=0; 
& g_{I}\left( s^{(\mathrm{dB})}_{0} \right)=\frac{1}{2}; 
 \frac{\mathrm{d} g_{I}\left( s^{(\mathrm{dB})} \right)}{\mathrm{d}s^{(\mathrm{dB})}}  \Bigg|_{s^{(\mathrm{dB})}=s^{(\mathrm{dB})}_{0}} = b_{0} 
\end{align}
\else
\begin{align}
& \underset{s^{(\mathrm{dB})}\to -\infty }{\mathop{\lim }}\,g_{I}\left( s^{(\mathrm{dB})} \right)=1;\underset{s^{(\mathrm{dB})}\to +\infty }{\mathop{\lim }}\,g_{I}\left( s^{(\mathrm{dB})} \right)=0 
\nonumber \\ 
& g_{I}\left( s^{(\mathrm{dB})}_{0} \right)=\frac{1}{2}; \quad
 \frac{\mathrm{d}}{\mathrm{d}s^{(\mathrm{dB})}} g_{I}\left( s^{(\mathrm{dB})} \right) \Bigg|_{s^{(\mathrm{dB})}=s^{(\mathrm{dB})}_{0}} = b_{0} 
\end{align}
\fi

These two parameters can be easily obtained from the properties of the sigmoidal logistic function given above solving the following equation: 
$
\mathcal{L}_{I}\left( s^{(\mathrm{dB})}_{0}|\mathcal{X}_{\mathrm{MT}_{0}}^{(j)} \right)=1/2
$
, which gives $s^{(\mathrm{dB})}_0$ and then
$
b_{0}=-4 \mathrm{d/d}s^{(\mathrm{dB})}\mathcal{L}_{I}\left( s^{(\mathrm{dB})}|\mathcal{X}_{\mathrm{MT}_{0}}^{(j)} \right)|_{s^{(\mathrm{dB})}=s^{(\mathrm{dB})}_{0}}
$; however, there exist a rich literature advocated to obtain those parameters efficiently which is called logistic regression \cite{Hosmer13}. In order to do that, it is only necessary to evaluate the Laplace transform of the interference given by Lemma \ref{prop:Soft IAFPC LI} for a few sample values $\{s^{(\mathrm{dB})}_1,s^{(\mathrm{dB})}_2,\cdots,s^{(\mathrm{dB})}_n\}$ and then perform logistic regression, which is a built in function available in common mathematical software packages like Mathematica or MALAB. This allows to obtain quickly the parameters $s^{(\mathrm{dB})}_0$ and $b_0$  avoiding the need to solve the aforementioned equation.
Hence the process for numerical evaluation of (\ref{eq:Soft IAFPC ccdf SINR}) consists on the following steps: (i) evaluate (\ref{eq:Soft IAFPC LI}) for $n$ sample points (good results are obtained with $n \sim 8$), (ii) perform logistic regression over the sample points so as to obtain $b_0$ and $s^{(\mathrm{dB})}_0$ and (iii) evaluate (\ref{eq:Soft IAFPC ccdf SINR}).
Finally, using $s$ in linear scale the Laplace transform is approximated as
\begin{equation}
\mathcal{L}_{{\hat{I}}}\left( s|\mathcal{X}_{\mathrm{MT}_{0}}^{(j)} \right)=\frac{1}{1+\mathrm{e}^{b_{0}\left( 10\log _{10}\left( s \right)-s_{0}^{(\mathrm{dB})} \right)}}
\end{equation}

\subsection{Transformed Distribution Approach}
\label{sec:Transformed Distribution Approach}

Another approach, which we name Transformed Distribution Approach (TDA), is to approximate the Laplace transform of the interference by a suitable function $f(s,\underline{\theta}^{(j)})$ defined with $n$ parameters $\underline{\theta}^{(j)}=\{\theta^{(j)}_0,\cdots,\theta^{(j)}_{n-1}\}$. Such function will be the Laplace transform of a particular distribution and thus we perform MM of $n$ moments in order to obtain the parameters that define such distribution. The matching between the proposed function and the Laplace transform of the interference only needs to be accurate for $s \in \mathbb{R}^{+}$, since it is only evaluated for positive values in order to obtain the ccdf of the SINR. The Laplace transform of the interference satisfies these two conditions $\mathcal{L}_{I}\left( 0|\mathcal{X}_{\mathrm{MT}_{0}}^{(j)} \right)=1$ and 
$\mathcal{L}_{I}\left( \infty|\mathcal{X}_{\mathrm{MT}_{0}}^{(j)} \right)=0$. Hence suitable functions must satisfy the following conditions:
\begin{align}
\label{eq:TDA cond1}
& \underset{s\to 0}{\mathop{\lim }}\,f\left( s,\underset{\raise0.3em\hbox{$\smash{\scriptscriptstyle-}$}}{\theta }^{(j)} \right)=1 
; \quad \underset{s\to \infty }{\mathop{\lim }}\,f\left( s,\underset{\raise0.3em\hbox{$\smash{\scriptscriptstyle-}$}}{\theta }^{(j)} \right)=0 
\\ 
\label{eq:TDA cond2}
& 0< \Bigg|\frac{\mathrm{d}^{r}}{\mathrm{d}s^{r}}f\left( s,\underset{\raise0.3em\hbox{$\smash{\scriptscriptstyle-}$}}{\theta }^{(j)} \right)\Big|_{s=0}\Bigg|<\infty, \; r\in [1,n] 
\end{align}
where (\ref{eq:TDA cond1}) is necessary to have the Laplace transform of a pdf and (\ref{eq:TDA cond2}) is necessary to perform MM over $n$ moments\footnote{We have restricted to have finite moments for the $n$ first moments, however it is only necessary to have $n$ finite moments to perform MM.}, since moments are obtained from derivatives of the Laplace transform evaluated in 0. Hence we approximate the interference $I|\mathcal{X}^{(j)}_{\mathrm{MT}_0}$ to $\hat{I}|\mathcal{X}^{(j)}_{\mathrm{MT}_0}$ being $\mathcal{L}_{\hat{I}}\left( s|\mathcal{X}_{\mathrm{MT}_{0}}^{(j)} \right)=f(s,\underline{\theta}^{(j)})$.
Following this approach we propose two suitable functions to approximate the Laplace transform of the interference.
\subsubsection{Exponential function}
As a suitable function we propose the following function
\begin{equation}
\label{eq:TDA exp}
f\left( s,\theta _{0}^{(j)} \right)=\operatorname{e}^{-\theta _{0}^{(j)}s}\mathbf{1}\left( s\ge 0 \right)
\end{equation}
This function satisfies conditions given with (\ref{eq:TDA cond1}) and (\ref{eq:TDA cond2}). Performing MM yields 
$\theta _{0}^{(j)}=-\beta '^{(j)}\left( 0 \right)$. 
\subsubsection{Algebraic function}
The following function also satisfies the aforementioned conditions.
\begin{equation}
f\left( s,\theta _{0}^{(j)} \right)=\frac{1}{1+s\theta _{0}^{(j)}}
\end{equation}
Performing MM also yields 
$\theta _{0}^{(j)}=-\beta '^{(j)}\left( 0 \right)$. 

Notice that approximating the Laplace transform with the approaches proposed in this section leads (for the ccdf of the SINR) to two nested integrals instead of four, hence the reduction in computational complexity is considerable. 

\section{Asymptotic Analysis}
\label{sec:Asymptotic Analysis}
In this section the obtained expressions are evaluated when $i_0$ tends to $\infty$ and when $i_0$ is low. The former case is interesting since it represents the performance of non IA FCP. The latter case illustrate the trend as as the maximum allowed interference level $i_0$ becomes lower. 

\subsection{Non Interference Aware Power Control ($i_0 \rightarrow \infty$)}
\label{sec:Non Interference Aware Power Control}
From the expression of the transmitted power given in (\ref{eq:pMT Soft}) it can be observed that interference awareness is lost if $i_0$ tends to $\infty$. Hence obtaining the limits as $i_0 \rightarrow \infty$ on the previously calculated expressions yields the performance of non IA FPC. In this case the transmit power is expressed as 
$
p_{\mathrm{MT}}^{\left( i_{0}\to \infty  \right)}\left( R_{\mathrm{MT}_{0}} \right)=\min \left( p_{0}\left( \tau R_{\mathrm{MT}_{0}} \right)^{\alpha \varepsilon },p_{\mathrm{max} } \right)
$.

\begin{proposition}
\label{prop:non IA E(Pmt)}
The average transmitted power with non IA FPC is given as

\ifOneColumn
\begin{align}
\label{eq:non IA E(Pmt)}
 P_{\mathrm{MT}_{0}}^{\left( i_{0}\to \infty  \right)} = \underset{i_{0}\to \infty }{\mathop{\lim }}\,\mathbb{E}\left[ P_{\mathrm{MT}_{0}} \right] 
 =\sum\limits_{j\in \mathcal{K}}{\Pr \left( \mathcal{X}_{\mathrm{MT}_{0}}^{(j)} \right)}
{\mathbb{E}}_{R_{\mathrm{MT}_{0}}}
\left[ p_{\mathrm{MT}}\left( R_{\mathrm{MT}_{0}} \right)|\mathcal{X}_{\mathrm{MT}_{0}}^{(j)} \right]  
\end{align}
\else
\begin{align}
\label{eq:non IA E(Pmt)}
& P_{\mathrm{MT}_{0}}^{\left( i_{0}\to \infty  \right)} = \underset{i_{0}\to \infty }{\mathop{\lim }}\,\mathbb{E}\left[ P_{\mathrm{MT}_{0}} \right] 
\nonumber \\ 
& =\sum\limits_{j\in \mathcal{K}}{\Pr \left( \mathcal{X}_{\mathrm{MT}_{0}}^{(j)} \right)}
{\mathbb{E}}_{R_{\mathrm{MT}_{0}}}
\left[ p_{\mathrm{MT}}\left( R_{\mathrm{MT}_{0}} \right)|\mathcal{X}_{\mathrm{MT}_{0}}^{(j)} \right]  
\end{align}
being
\fi

\begin{equation}
\label{eq:non IA f_R}
f_{R_{\mathrm{MT}_{0}}}\left( v|\mathcal{X}_{\mathrm{MT}_{0}}^{(j)} \right)=\frac{f_{R_{\mathrm{MT}_{0},(1)}^{(j)}}\left( v \right)\cdot \bar{F}_{R_{\mathrm{MT}_{0},(1)}^{(\tilde{j})}}\left( \left( \frac{t^{(\tilde{j})}}{t^{(j)}} \right)^{\frac{1}{\alpha }}v \right)}{\Pr \left( \mathcal{X}_{\mathrm{MT}_{0}}^{(j)} \right)}
\end{equation}
\end{proposition}

\begin{proof}
Taking the limit when $i_0 \to \infty$ the transmit power does not depend on $U_{\mathrm{MT}_0}$. Hence we have

\ifOneColumn
\begin{align}
& \mathbb{E}\left[ P_{\mathrm{MT}_{0}}^{\left( i_{0}\to \infty  \right)} \right]=\sum\limits_{j\in \mathcal{K}}{\int\limits_{v=0}^{\infty }{\sum\limits_{m\in \mathcal{K}}{\Pr \left( \mathcal{X}_{\mathrm{MT}_{0}}^{(j,m)} \right)}}} 
 \int\limits_{w=0}^{\infty }{p_{\mathrm{MT}}\left( R_{\mathrm{MT}_{0}} \right)f_{R_{\mathrm{MT}_{0}},U_{\mathrm{MT}_{0}}}\left( v,w|\mathcal{X}_{\mathrm{MT}_{0}}^{(j,m)} \right)}\cdot \mathrm{d}w 
\nonumber \\ 
& =\sum\limits_{j\in \mathcal{K}}{\Pr \left( \mathcal{X}_{\mathrm{MT}_{0}}^{(j)} \right)}\mathbb{E}_{R_{\mathrm{MT}_{0}}}\left[ p_{\mathrm{MT}}\left( R_{\mathrm{MT}_{0}} \right)|\mathcal{X}_{\mathrm{MT}_{0}}^{(j)} \right] 
\end{align}
\else
\begin{align}
& \mathbb{E}\left[ P_{\mathrm{MT}_{0}}^{\left( i_{0}\to \infty  \right)} \right]=\sum\limits_{j\in \mathcal{K}}{\int\limits_{v=0}^{\infty }{\sum\limits_{m\in \mathcal{K}}{\Pr \left( \mathcal{X}_{\mathrm{MT}_{0}}^{(j,m)} \right)}}} 
\nonumber \\ 
& \int\limits_{w=0}^{\infty }{p_{\mathrm{MT}}\left( R_{\mathrm{MT}_{0}} \right)f_{R_{\mathrm{MT}_{0}},U_{\mathrm{MT}_{0}}}\left( v,w|\mathcal{X}_{\mathrm{MT}_{0}}^{(j,m)} \right)}\cdot \mathrm{d}w 
\nonumber \\ 
& =\sum\limits_{j\in \mathcal{K}}{\Pr \left( \mathcal{X}_{\mathrm{MT}_{0}}^{(j)} \right)}\mathbb{E}_{R_{\mathrm{MT}_{0}}}\left[ p_{\mathrm{MT}}\left( R_{\mathrm{MT}_{0}} \right)|\mathcal{X}_{\mathrm{MT}_{0}}^{(j)} \right] 
\end{align}
\fi
\noindent where it has been used the following fact that comes after applying the total probability theorem.

\ifOneColumn
\begin{align}
\label{eq:marginal pdf}
& \sum\limits_{m\in \mathcal{K}}{\Pr \left( \mathcal{X}_{\mathrm{MT}_{0}}^{(j,m)} \right)}\int\limits_{w=0}^{\infty }{f_{R_{\mathrm{MT}_{0}},U_{\mathrm{MT}_{0}}}\left( v,w|\mathcal{X}_{\mathrm{MT}_{0}}^{(j,m)} \right)}\cdot \mathrm{d}w 
 =\Pr \left( \mathcal{X}_{\mathrm{MT}_{0}}^{(j)} \right)f_{R_{\mathrm{MT}_{0}}}\left( v|\mathcal{X}_{\mathrm{MT}_{0}}^{(j)} \right) 
\end{align}
\else
\begin{align}
\label{eq:marginal pdf}
& \sum\limits_{m\in \mathcal{K}}{\Pr \left( \mathcal{X}_{\mathrm{MT}_{0}}^{(j,m)} \right)}\int\limits_{w=0}^{\infty }{f_{R_{\mathrm{MT}_{0}},U_{\mathrm{MT}_{0}}}\left( v,w|\mathcal{X}_{\mathrm{MT}_{0}}^{(j,m)} \right)}\cdot \mathrm{d}w 
\nonumber \\ 
& =\Pr \left( \mathcal{X}_{\mathrm{MT}_{0}}^{(j)} \right)f_{R_{\mathrm{MT}_{0}}}\left( v|\mathcal{X}_{\mathrm{MT}_{0}}^{(j)} \right) 
\end{align}
\fi

Finally, solving the above expression after integrating over $w$ yields the marginal pdf given in (\ref{eq:non IA f_R}).
\end{proof}

\begin{lemma}
\label{prop:non IA LI}
The Laplace transform of the interference when $i_0 \to \infty$ is 

\ifOneColumn
\begin{align}
\label{eq:non IA}
& \mathcal{L}_{I}^{\left( i_{0}\to \infty  \right)}\left( s|\mathcal{X}_{\mathrm{MT}_{0}}^{(j)} \right)
= \exp \Big(-\sum\limits_{k\in \mathcal{K}}{2\pi \lambda ^{(k)}} 
 \times \int\limits_{r=0}^{\infty }{\chi ^{\left( i_{0}\to \infty  \right)}\left( s,r \right)}f_{R_{\mathrm{MT}_{i}}}\left( r|\mathcal{X}_{\mathrm{MT}_{i}}^{(k)} \right)\mathrm{d}r \Big) 
\end{align}
\else
\begin{align}
\label{eq:non IA}
& \mathcal{L}_{I}^{\left( i_{0}\to \infty  \right)}\left( s|\mathcal{X}_{\mathrm{MT}_{0}}^{(j)} \right)
= \exp \Big(-\sum\limits_{k\in \mathcal{K}}{2\pi \lambda ^{(k)}} 
\nonumber \\ 
& \times \int\limits_{r=0}^{\infty }{\chi ^{\left( i_{0}\to \infty  \right)}\left( s,r \right)}f_{R_{\mathrm{MT}_{i}}}\left( r|\mathcal{X}_{\mathrm{MT}_{i}}^{(k)} \right)\mathrm{d}r \Big) 
\end{align}
\fi
\noindent with 
\begin{align}
& \chi ^{\left( i_{0}\to \infty  \right)}\left( s,r \right)=\frac{r^{2-\alpha } s p_{\mathrm{MT}}^{\left( i_{0}\to \infty  \right)}\left( r \right)\tau ^{-\alpha }}{\alpha -2}\left( \frac{t^{(j)}}{t^{(k)}} \right)^{\frac{2-\alpha }{\alpha }}\, 
\nonumber \\ 
& _{2}F_{1}\left( 1,\frac{\alpha -2}{\alpha };2-\frac{2}{\alpha };-\frac{sp_{\mathrm{MT}}^{\left( i_{0}\to \infty  \right)}\left( r \right)\tau ^{-\alpha }}{r}\left( \frac{t^{(j)}}{t^{(k)}} \right)^{-1} \right) 
\end{align}

\end{lemma}

\begin{proof}
The limit of the Laplace transform can be expressed as 
$
\mathcal{L}_{I}^{\left( i_{0}\to \infty  \right)}\left( s|\mathcal{X}_{\mathrm{MT}_{0}}^{(j)} \right)=\mathrm{e}^{ \left( \beta ^{\left( i_{0}\to \infty  \right)(j)}\left( s \right) \right)}
$. Hence the limit of $\beta ^{(j)}\left( s \right)$ appears below

\begin{align}
& \beta ^{\left( i_{0}\to \infty  \right)(j)}\left( s \right)
=-\sum\limits_{k\in \mathcal{K}}{2\pi \lambda ^{(k)}}\times \int\limits_{r=0}^{\infty }{\chi ^{\left( i_{0}\to \infty  \right)}\left( s,r \right)} 
\nonumber \\ 
& \sum\limits_{n\in \mathcal{K}}{\frac{\Pr \left( \mathcal{X}_{\mathrm{MT}_{i}}^{(k,n)} \right)}{\Pr \left( \mathcal{X}_{\mathrm{MT}_{i}}^{(k)} \right)}}\int\limits_{u=\left( \frac{t^{(n)}}{t^{(k)}} \right)^{\frac{1}{\alpha }}r}^{\infty }{f_{R_{\mathrm{MT}_{i}},U_{\mathrm{MT}_{i}}}\left( r,u|\mathcal{X}_{\mathrm{MT}_{i}}^{(k,n)} \right)}\mathrm{d}r\mathrm{d}u 
\end{align}
\noindent where it has been used again the fact given in (\ref{eq:marginal pdf}).
\end{proof}
The ccdf of the SINR can be obtained taking the limit from (\ref{eq:Soft IAFPC ccdf SINR}). 

\ifOneColumn
\begin{align}
& \bar{F}_{\text{SINR}}^{\left( i_{0}\to \infty  \right)}\left( \gamma  \right)=\,\sum\limits_{j\in \mathcal{K}}{\int\limits_{v=0}^{\infty }{e^{-\frac{\gamma \sigma _{n}^{2}\left( \tau v \right)^{\alpha }}{p_{\text{MT}}\left( v \right)}}}f_{R_{\text{MT}_{0}}}\left( v|\mathcal{X}_{\text{MT}_{0}}^{(j)} \right)} 
 \exp \Bigg(-\sum\limits_{k\in \mathcal{K}}{2\pi \lambda ^{(k)}}\int\limits_{r=0}^{\infty }{f_{R_{\text{MT}_{i}}}\left( r|\mathcal{X}_{\text{MT}_{i}}^{(k)} \right)} 
\nonumber \\ 
 & \chi ^{\left( i_{0}\to \infty  \right)}\left( \frac{\gamma \sigma _{n}^{2}\left( \tau v \right)^{\alpha }}{p_{\text{MT}}\left( v \right)},r \right)\text{d}r \Bigg)\cdot \text{d}v 
\end{align}
\else
\begin{align}
& \bar{F}_{\text{SINR}}^{\left( i_{0}\to \infty  \right)}\left( \gamma  \right)=\,\sum\limits_{j\in \mathcal{K}}{\int\limits_{v=0}^{\infty }{e^{-\frac{\gamma \sigma _{n}^{2}\left( \tau v \right)^{\alpha }}{p_{\text{MT}}\left( v \right)}}}f_{R_{\text{MT}_{0}}}\left( v|\mathcal{X}_{\text{MT}_{0}}^{(j)} \right)} 
\nonumber \\ 
 & \exp \Bigg(-\sum\limits_{k\in \mathcal{K}}{2\pi \lambda ^{(k)}}\int\limits_{r=0}^{\infty }{f_{R_{\text{MT}_{i}}}\left( r|\mathcal{X}_{\text{MT}_{i}}^{(k)} \right)} 
\nonumber \\ 
 & \chi ^{\left( i_{0}\to \infty  \right)}\left( \frac{\gamma \sigma _{n}^{2}\left( \tau v \right)^{\alpha }}{p_{\text{MT}}\left( v \right)},r \right)\text{d}r \Bigg)\cdot \text{d}v 
\end{align}
\fi
\noindent where it has been used the fact that when $i_0 \to \infty$ the Laplace transform and transmit power law do not depend on $w$ and then the result has been integrated and summed over $w$ and $k$, respectively. It can be observed that the ccdf of the SINR has been simplified to a two fold integral for the non IA case.

\subsection{Low $i_0$ regime ($i_0 \sim 0$)}
\label{sec:Low i_0 regime}
In this section we obtain expressions that are valid as $i_0$ is low. Following the notation proposed in Section \ref{sec:Introduction} in the low $i_0$ regime the transmit power can be expressed as
$
p_{\mathrm{MT}}^{\left( i_{0} \sim 0 \right)}\left( U_{\mathrm{MT}_{0}} \right)=i_{0}\left( \tau U_{\mathrm{MT}_{0}} \right)^{\alpha }
$
since 
$
\underset{i_{0}\to 0}{\mathop{\lim }}\,\frac{p_{\mathrm{MT}}\left( R_{\mathrm{MT}_{0}},U_{\mathrm{MT}_{0}} \right)}{i_{0}\left( \tau U_{\mathrm{MT}_{0}} \right)^{\alpha }}=1
$.

\subsubsection{Average Transmit Power}

The average transmit power in the low $i_0$ regime can be expressed as follows

\begin{equation}
\mathbb{E}\left[ P_{\mathrm{MT}_{0}}^{\left( i_{0} \sim 0 \right)} \right]=\sum\limits_{m\in \mathcal{K}}{\Pr \left( \mathcal{Q}_{\mathrm{MT}_{0}}^{(m)} \right)}\mathbb{E}_{U_{\mathrm{MT}_{0}}}\left[ i_{0}\left( \tau U_{\mathrm{MT}_{0}} \right)^{\alpha }|\mathcal{Q}_{\mathrm{MT}_{0}}^{(m)} \right]
\end{equation}
\noindent since in this case the transmit power only depends on $U_{\mathrm{MT}_{0}}$. This marginal pdf can be obtained from (\ref{prop:Soft IAFPC f_RMT_UMT Cond Xjm}) as

\ifOneColumn
\begin{align}
& f_{U_{\mathrm{MT}_{0}}}\left( w|\mathcal{Q}_{\mathrm{MT}_{0}}^{(m)} \right)= 
 \frac{\sum\limits_{j\in \mathcal{K}}{\Pr \left( \mathcal{X}_{\mathrm{MT}_{0}}^{(j,m)} \right)}}{\Pr \left( \mathcal{Q}_{\mathrm{MT}_{0}}^{(m)} \right)}\int\limits_{v=0}^{\infty }{f_{R_{\mathrm{MT}_{0}},U_{\mathrm{MT}_{0}}}\left( v,w|\mathcal{X}_{\mathrm{MT}_{0}}^{(j,m)} \right)}\cdot \mathrm{d}v 
\end{align}
\else
\begin{align}
  & f_{U_{\mathrm{MT}_{0}}}\left( w|\mathcal{Q}_{\mathrm{MT}_{0}}^{(m)} \right)= 
  \nonumber \\
 & \frac{\sum\limits_{j\in \mathcal{K}}{\Pr \left( \mathcal{X}_{\mathrm{MT}_{0}}^{(j,m)} \right)}}{\Pr \left( \mathcal{Q}_{\mathrm{MT}_{0}}^{(m)} \right)}\int\limits_{v=0}^{\infty }{f_{R_{\mathrm{MT}_{0}},U_{\mathrm{MT}_{0}}}\left( v,w|\mathcal{X}_{\mathrm{MT}_{0}}^{(j,m)} \right)}\cdot \mathrm{d}v 
\end{align}
\fi

\emph{Special case: min path loss}. In case of minimum path loss association criteria all the association weights $t^{(j)}$ are equal to 1. Minimum path loss association is equivalent to single tier (since in our model we have the same path loss per tier) \cite{Singh15} (Corollary 6) where the density of single tier BSs is 
$\lambda =\sum\limits_{j\in \mathcal{K}}{\lambda ^{(j)}}$. With minimum path loss association criteria the distance towards the serving BS is the distance to the nearest BS and the distance to the most interfered BS is the distance to the second nearest BS. Hence we have the following equivalence
$
f_{R_{\mathrm{MT}_{0}},U_{\mathrm{MT}_{0}}}\left( v,w \right) =
4\left( \pi  \lambda  \right)^{2} v  w  \mathrm{e}^{-\pi \lambda w^{2}} \mathbf{1} (v < w)
$. 
In this case the average transmit power can be further simplified to

\begin{equation}
\mathbb{E}\left[ P_{\mathrm{MT}_{0}}^{\left( i_{0}\sim 0 \right)} \right]=i_{0}\left( \frac{\tau }{\sqrt{\pi \lambda }} \right)^{\alpha }\Gamma \left( 2+\frac{\alpha }{2} \right)
\end{equation}

\subsubsection{Laplace Transform of the Interference}
The Laplace transform of the interference can be expressed as

\ifOneColumn
\begin{align}
& \mathcal{L}_{I}^{\left( i_{0} \sim 0 \right)}\left( s|\mathcal{X}_{\mathrm{MT}_{0}}^{(j)} \right) 
  = \exp \Bigg( -\sum\limits_{k\in \mathcal{K}}{2\pi \lambda ^{(k)}\sum\limits_{n\in \mathcal{K}}{\Pr \left( \mathcal{Q}_{\mathrm{MT}_{i}}^{(n)}|\mathcal{X}_{\mathrm{MT}_{i}}^{(k)} \right)}}  
\int\limits_{r=0}^{\infty }\int\limits_{u=\left( \frac{t^{(n)}}{t^{(k)}} \right)^{\frac{1}{\alpha }}r}^{\infty }
\nonumber  \\ 
& f_{R_{\mathrm{MT}_{i}},U_{\mathrm{MT}_{i}}}\left( r,u|\mathcal{X}_{\mathrm{MT}_{i}}^{(k,n)} \right)
\chi ^{\left( i_{0} \sim 0 \right)}\left( s,r,u \right) \mathrm{d}r \mathrm{d}u \Bigg) 
\end{align}
\else
\begin{align}
& \mathcal{L}_{I}^{\left( i_{0} \sim 0 \right)}\left( s|\mathcal{X}_{\mathrm{MT}_{0}}^{(j)} \right) 
\nonumber  \\ 
&  = \exp \Bigg( -\sum\limits_{k\in \mathcal{K}}{2\pi \lambda ^{(k)}\sum\limits_{n\in \mathcal{K}}{\Pr \left( \mathcal{Q}_{\mathrm{MT}_{i}}^{(n)}|\mathcal{X}_{\mathrm{MT}_{i}}^{(k)} \right)}}  
\int\limits_{r=0}^{\infty }\int\limits_{u=\left( \frac{t^{(n)}}{t^{(k)}} \right)^{\frac{1}{\alpha }}r}^{\infty }
\nonumber  \\ 
&f_{R_{\mathrm{MT}_{i}},U_{\mathrm{MT}_{i}}}\left( r,u|\mathcal{X}_{\mathrm{MT}_{i}}^{(k,n)} \right)
\chi ^{\left( i_{0} \sim 0 \right)}\left( s,r,u \right) \mathrm{d}r \mathrm{d}u \Bigg) 
\end{align}
\fi
\noindent where the term $\chi ^{\left( i_{0} \sim 0 \right)}\left( s,r,u \right)$ can be written as
\begin{align}
\label{eq:Low i0 general association LTI}
& \chi ^{\left( i_{0} \sim 0 \right)}\left( s,r,u \right)=\frac{si_{0}\left( \tau u \right)^{\alpha }\tau ^{-\alpha }}{\alpha -2}\mathrm{max} ^{2-\alpha }\left( \left( \frac{t^{(j)}}{t^{(k)}} \right)^{\frac{1}{\alpha }}r,u \right) 
\nonumber \\ 
& \,_{2}F_{1}\left( 1,\frac{\alpha -2}{\alpha };2-\frac{2}{\alpha };\frac{-si_{0}\left( \tau u \right)^{\alpha }\tau ^{-\alpha }}{\mathrm{max}^{\alpha }\left( \left( \frac{t^{(j)}}{t^{(k)}} \right)^{\frac{1}{\alpha }}r,u \right)} \right) 
\end{align}

\emph{Special case: min path loss}. In case of minimum path loss association the term $\chi ^{\left( i_{0} \sim 0 \right)}\left( s,r,u \right)$ can be written as

\ifOneColumn
\begin{align}
\label{eq:Low i0 general association chi}
& \chi ^{\left( i_{0} \sim 0 \right)}\left( s,r,u \right)
 =\frac{\mathbf{1}\left( r\le u \right) si_{0}u^{2}}{\alpha -2}
 \,_{2}F_{1}\left( 1,\frac{\alpha -2}{\alpha };2-\frac{2}{\alpha };
 -si_{0} \right) + \,\frac{\mathbf{1}\left( r>u \right) si_{0}u^{\alpha } r^{2-\alpha }}{\alpha -2}
 \nonumber \\
& \,
 _{2}F_{1}\left( 1,\frac{\alpha -2}{\alpha };2-\frac{2}{\alpha };-si_{0}\left( \frac{u}{r} \right)
 ^{\alpha } \right) 
\end{align}
\else
\begin{align}
\label{eq:Low i0 general association chi}
& \chi ^{\left( i_{0} \sim 0 \right)}\left( s,r,u \right)
 =\frac{\mathbf{1}\left( r\le u \right) si_{0}u^{2}}{\alpha -2}
 \nonumber \\
& \,_{2}F_{1}\left( 1,\frac{\alpha -2}{\alpha };2-\frac{2}{\alpha };
 -si_{0} \right) + \,\frac{\mathbf{1}\left( r>u \right) si_{0}u^{\alpha } r^{2-\alpha }}{\alpha -2}
 \nonumber \\
& \,
 _{2}F_{1}\left( 1,\frac{\alpha -2}{\alpha };2-\frac{2}{\alpha };-si_{0}\left( \frac{u}{r} \right)
 ^{\alpha } \right) 
\end{align}
\fi
\noindent where we have particularized to equal association weights in (\ref{eq:Low i0 general association chi}) and then expressed the $\max(\cdot)$ function as the sum of two indicators functions. 
From the integration limits in (\ref{eq:Low i0 general association LTI}) it can be observed that with minimum path loss association the second term is canceled out. Finally integrating the resulting expression along $r$ and $u$ in (\ref{eq:Low i0 general association LTI}) yields

\begin{equation}
\label{eq:Low i0 min path LTI}
\mathcal{L}_{I}^{\left( i_{0} \sim 0 \right)}\left( s \right)=\exp \left( -\frac{4si_{0}}{\alpha -2}_{2}F_{1}\left( 1,\frac{\alpha -2}{\alpha };2-\frac{2}{\alpha };-si_{0} \right) \right)
\end{equation}
\noindent which is a closed form expression. 
\begin{remark}
Interestingly, it can be observed from (\ref{eq:Low i0 min path LTI}) that the Laplace transform of the interference in the low $i_0$ regime does not depend on the BS density. Hence all statistics of the interference like mean and variance do not depend on  $\lambda$.
\end{remark}
Using the above expression the ccdf of the SINR for the minimum path loss association case can be expressed as follows

\ifOneColumn
\begin{align}
\label{eq:Low i0 min path ccdf SINR}
& \bar{F}_{\text{SINR}}^{\left( i_{0}\sim 0 \right)}\left( \gamma  \right) = \int\limits_{v=0}^{\infty }{\int\limits_{w=v}^{\infty }{f_{R_{\text{MT}_{0}},U_{\text{MT}_{0}}}\left( v,w \right)}}\cdot \text{e}^{-\frac{\gamma \sigma _{n}^{2}}{i_{0}}\left( \frac{v}{w} \right)^{\alpha }} 
  \exp \Bigg(-_{2}F_{1}\left( 1,\frac{\alpha -2}{\alpha };2-\frac{2}{\alpha };-\gamma \left( \frac{v}{w} \right)^{\alpha } \right) 
\nonumber \\ 
& \frac{4\gamma }{\alpha -2}\left( \frac{v}{w} \right)^{\alpha } \Bigg)\text{d}v\text{d}w 
\end{align}
\else
\begin{align}
\label{eq:Low i0 min path ccdf SINR}
& \bar{F}_{\text{SINR}}^{\left( i_{0}\sim 0 \right)}\left( \gamma  \right) = \int\limits_{v=0}^{\infty }{\int\limits_{w=v}^{\infty }{f_{R_{\text{MT}_{0}},U_{\text{MT}_{0}}}\left( v,w \right)}}\cdot \text{e}^{-\frac{\gamma \sigma _{n}^{2}}{i_{0}}\left( \frac{v}{w} \right)^{\alpha }} 
\nonumber \\ 
 & \exp \Bigg(-_{2}F_{1}\left( 1,\frac{\alpha -2}{\alpha };2-\frac{2}{\alpha };-\gamma \left( \frac{v}{w} \right)^{\alpha } \right) 
\nonumber \\ 
 & \frac{4\gamma }{\alpha -2}\left( \frac{v}{w} \right)^{\alpha } \Bigg)\text{d}v\text{d}w 
\end{align}
\fi

\begin{remark}
By inspection of (\ref{eq:Low i0 min path ccdf SINR}) it can be noticed that although the Laplace transform of the interference does not depend on $\lambda$ the ccdf of the SINR does depend on $\lambda$ thanks to the joint pdf of distances. Interestingly, the dependence with $i_0$ only appears in the term 
$
\text{e}^{-\frac{\gamma \sigma _{n}^{2}}{i_{0}}\left( \frac{v}{w} \right)^{\alpha }}
$. Since this term disappears in the ideal no noise case ($\sigma _{n}^{2}=0$), in such case the dependence with $i_0$ is removed. 
\end{remark}

\section{Numerical Results}
\label{sec:Numerical Results}

\begin{table}
\renewcommand{\arraystretch}{1.3}
\caption{Simulation Parameters}
\label{tab:Simulation Parameters}
\centering
\begin{tabular}{ c c c c }
\toprule
Parameter & Value & Parameter & Value \\
\toprule
$f_c$ (MHz) & $2 \times 10^{3}$ & $h_{\mathrm{BS}}$ (m) & $10$ \\
\hline
$b_w$ (MHz) & $9$ & $t^{(1)}/t^{(2)}$ (dB) & $9$ \\
\hline
$\lambda^{(1)}$ (points/m$^2$) & $2 \times 10^{-6}$ & $\lambda^{(2)}$ (points/m$^2$) & $4 \times 10^{-6}$ \\
\hline
$\lambda_{\mathrm{MT}}$ (points/m$^2$) & $80 \times 10^{-6}$ & $n_\mathrm{thermal}$ (dBm/Hz) & $-174$  \\
\hline
$n_\mathrm{F}$ (dB) & $9$ & $\sigma_s$ (dB) & $4$  \\
\hline
$p_0$ (dBm) & $-70$ & $p_\mathrm{max}$ (dBm) & $\{\infty,5 \}$  \\
\hline
$i_0$ (dBm) & $[-120,-60 ]$ & $\epsilon$ & $[0,1]$  \\
\bottomrule
\end{tabular}
\end{table}

In this section theoretical expressions are evaluated numerically and compared with simulation results. It is considered a thermal noise power spectral density of $n_\mathrm{thermal}=-174$ dBm/Hz and a noise figure at the receiver of $n_\mathrm{F}=9$ dB. Main parameters are presented in Table \ref{tab:Simulation Parameters} unless otherwise stated.
Monte Carlo simulations are carried out to compare with theoretical results. 

Each simulation consists on $10^{4}$ spatial realizations that are averaged in order to obtain performance results. In simulation we consider the actual PP of interfering MT locations.

\subsection{Average Transmitted Power, Mean and Variance of the Interference}
\label{sec:E(P), E(I) and var(I)}

\begin{figure}[t]
\centering
\includegraphics[width=3in]{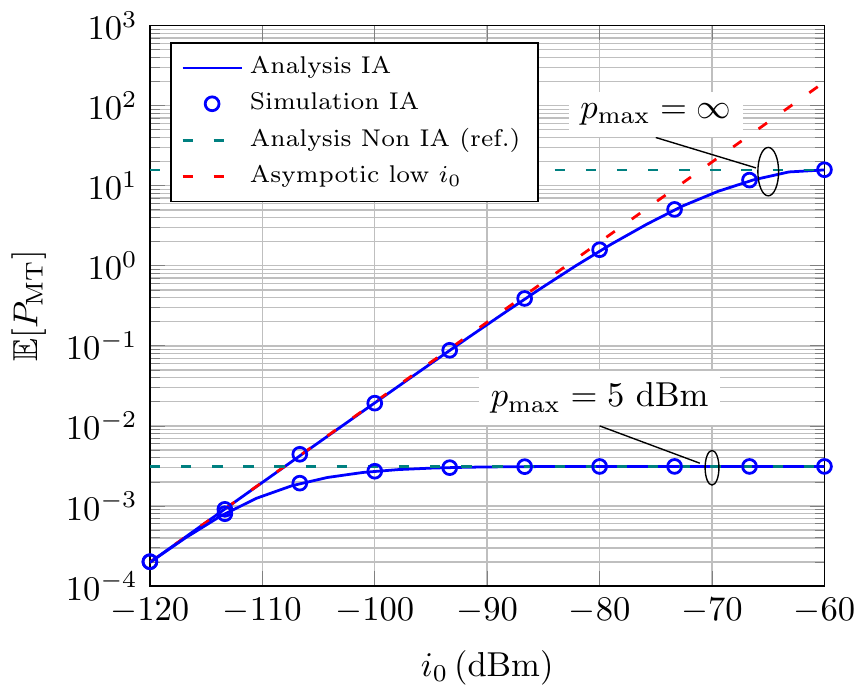}
\caption{Average transmitted power versus $i_0$ for IAFPC and non IA with $p_\mathrm{max} \rightarrow \infty$ and $p_\mathrm{max} = 5$ dBm.}
\label{fig:avP}
\end{figure}
%

\begin{figure*}[t]
\centering
\includegraphics[width=6.5in]{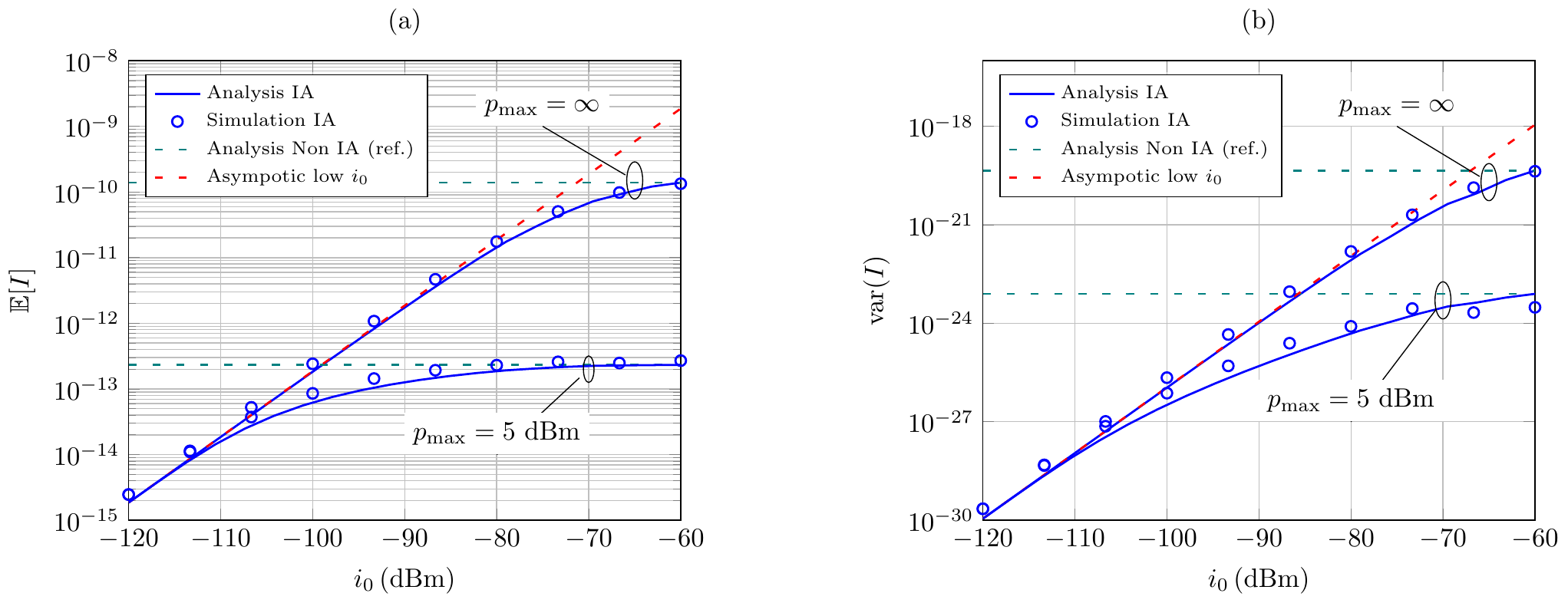}
\caption{(a) Mean of the interference versus $i_0$ for IAFPC and non IA with $p_\mathrm{max} \rightarrow \infty$ and $p_\mathrm{max} = 5$ dBm. (b) Variance of the interference versus $i_0$ for IAFPC and non IA with $p_\mathrm{max} \rightarrow \infty$ and $p_\mathrm{max} = 5$ dBm.}
\label{fig:avI}
\end{figure*}
%



In this section the average transmitted power, average interference and variance of the interference is obtained and compared with simulations. It is important to remark that the interfering MTs positions in the UL does not follow a PPP \cite{Singh15} and hence the interference term given by (\ref{eq:I}) is an approximation that aims to capture some correlations among the interfering MT positions and the positions of the probe BS. Therefore, the analysis of metrics related to interference (like mean and variance of the interference or ccdf of the SINR) represents an approximation. However, theoretical results that do not involve the interference (like the average transmitted power) are exact. This can be observed from Fig. \ref{fig:avP} and Fig. \ref{fig:avI} (a)-(b).

In Fig. \ref{fig:avP} the average transmitted power versus $i_0$ for IAFPC and non IA methods is presented considering $p_\mathrm{max} \rightarrow \infty$ and $p_\mathrm{max} = 5$ dBm. Since $i_0$ does not exit in non IA FPC, the value obtained with this method is actually a reference value which is drawn to facilitate the comparison. It can be observed that both IA and non IA methods provide the same average transmitted power for $i_0=-60$ dBm. This is reasonable since when $i_0$ is high enough almost all MTs do not truncate their transmissions due to $i_0$. It can be also contrasted that results from asymptotic analysis (red curve) are only accurate for low values of $i_0$ as it was expected. The same holds for the mean and variance of the interference in Fig. \ref{fig:avI} (a) and (b).
The variance of the interference degrades the estimation of the SINR,  thus affecting negatively to the performance of AMC in real implementations \cite{Zhang12}. Hence, methods that reduce the variance of the interference are appealing. 
From figures presented in this section it can be observed that IAFPC reduce both the mean and variance of the interference and the average transmitted power compared to non IA FPC which is highly beneficial. 

\subsection{Approximating the Interference in IAFPC}
\label{sec:Approximating the interference in IAFPC}

\begin{figure}[t]
\centering
\includegraphics[width=3in]{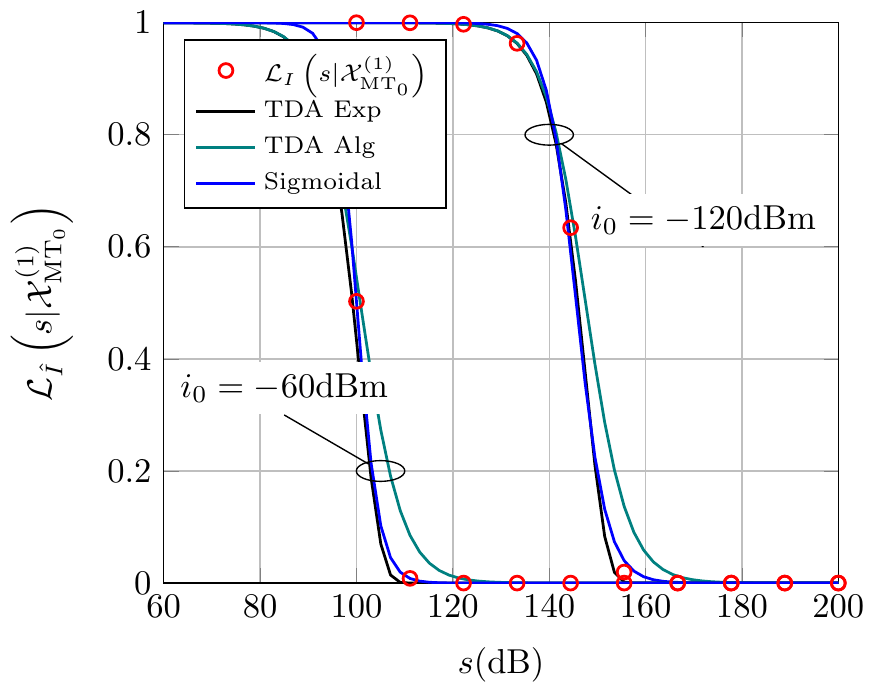}
\caption{Laplace transform of the analytic interference term $I$ given by (\ref{eq:I}) conditioned on $\mathcal{X}^{(1)}_{\mathrm{MT}_0}$ for IAFPC and its approximation by $\hat{I}$.}
\label{fig:LTI_X1}
\end{figure}

As it has been explained in Sections \ref{sec:System model} and \ref{sec:Statistical Modeling of the interference}, the exact interference term is intractable since the positions of interfering MTs do not follow a PPP and thus an approximation to the interference is proposed with (\ref{eq:I}). Although this analytic interference term $I$ leads to tractable expressions, for IAFPC the numeric complexity of obtained expressions require further approximations. Hence, in Section \ref{sec:Statistical Modeling of the interference} two approaches to approximate the interference term $I$ by $\hat{I}$ were proposed. The aim of this section is to compare such approaches in terms of the Laplace transform of its pdf. 

Fig. \ref{fig:LTI_X1} illustrates the Laplace transform of $\hat{I}$ with Sigmoidal  and TDA approximations conditioned on being the probe MT associated to tier 1. The s-axis is expressed in dBs and hence it is possible to observe the s-shape of the Laplace transform. The logistic regression for Sigmoidal approximation has considered 8 equally spaced points between $s=60$ dBs and $s=200$ dBs. It can be noticed that the difference in the Laplace transform of the interference between $i_0=-60$ dBm and $i_0=-120$ dBm is high and such interference tend to reach higher values for $i_0=-60$ dBm. All approximations follow the same trend as the Laplace transform of the analytic interference $I$. 

\subsection{ccdf of the SINR}
\label{sec:ccdf of the SINR}

\begin{figure*}[t]
\centering
\includegraphics[width=6.5in]{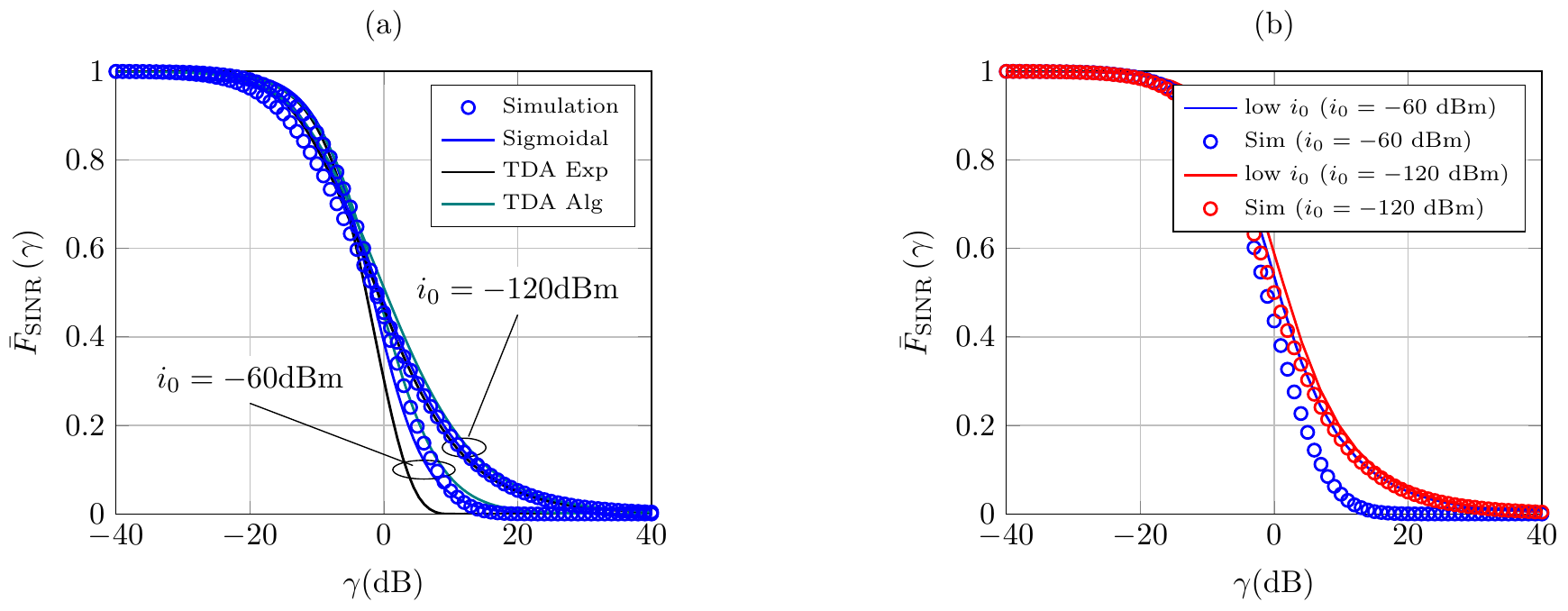}
\caption{(a) ccdf of the SINR for the typical MT using IAFPC method with $\epsilon=1$ and $i_0=\{-120, -60\}$ dBm. (b) ccdf of the SINR in the low $i_0$ regime given by (\ref{eq:Low i0 min path ccdf SINR}) using minimum path loss association.}
\label{fig:comparativa_sia_epsilon1}
\end{figure*}


%
%
%

This section illustrates the ccdf of the SINR IAFPC and compares analytic and simulation results considering $p_\mathrm{max} \rightarrow \infty$. Fig. \ref{fig:comparativa_sia_epsilon1} (a) illustrates the ccdf of the SINR associated with IAFPC with $\epsilon=1$. 
It can be observed that Sigmoidal and TDA with algebraic function represents a better approximation to ccdf of the SINR obtained by simulations than approximation with TDA using an exponential function for $i_0=-120$ and $i_0=-60$ dBm. 

Finally, Fig. \ref{fig:comparativa_sia_epsilon1} (b) illustrates a comparison between the ccdf of the SINR in the low $i_0$ regime given by (\ref{eq:Low i0 min path ccdf SINR}) and simulation results using minimum path loss association. Simulation is carried out for $\epsilon=1$. It can be observed a good math between simulation and asymptotic analysis for $i_0=-120$ dBm; however this matching is severely reduced for $i_0=-60$ dBm as it was expected. 

\subsection{Spectral Efficiency}
\label{sec:Spectral Efficiency}

\begin{figure*}[t]
\centering
\includegraphics[width=6.5in]{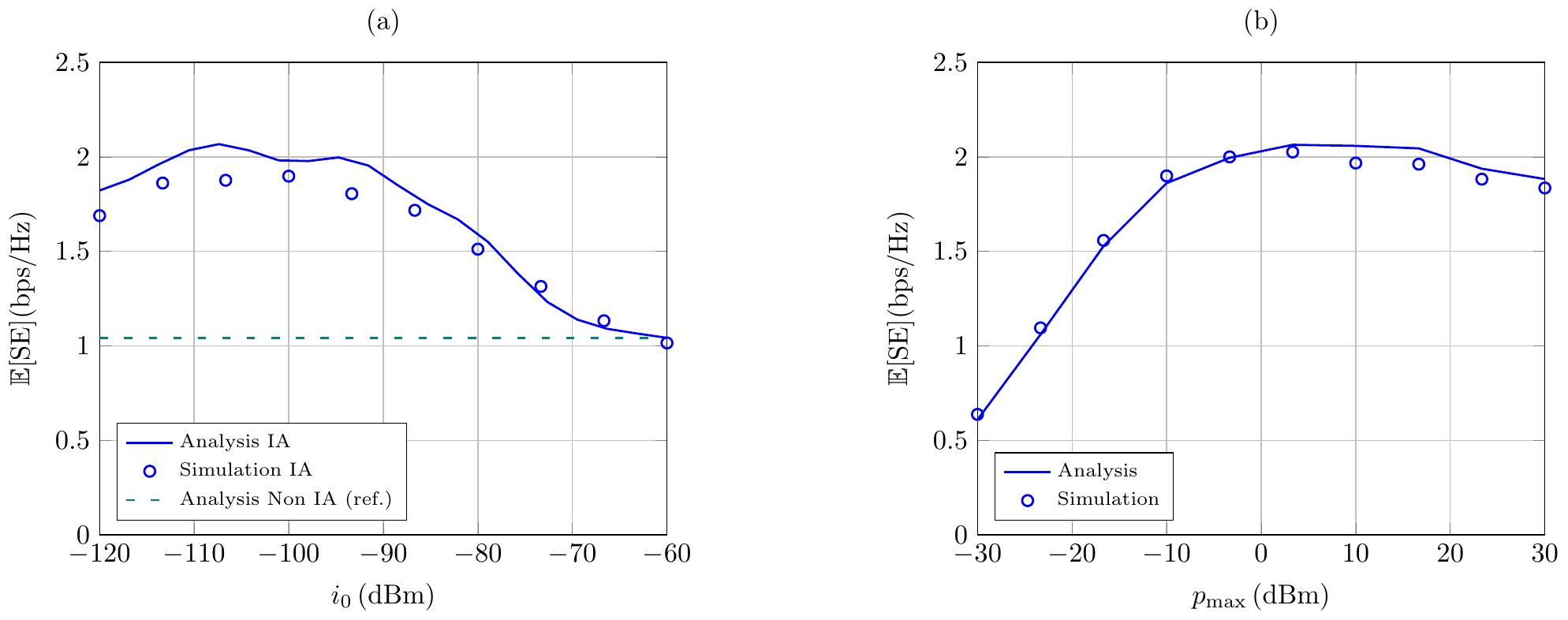}
\caption{(a) Average SE versus $i_0$ for IAFPC with $\epsilon=1$ and $p_{\mathrm{max}} \rightarrow \infty$. (b) Average SE versus $p_{\mathrm{max}}$ for IAFPC with $\epsilon=1$ and $i_0 = 90$ dBm.}
\label{fig:avSE_i0}
\end{figure*}


%

%
%

In this section the average SE is compared for IA and non IA methods. Such metric represents how well a MT's transmission exploits the bandwidth and it is expressed in terms of bits per second per Hertz (bps/Hz). For the IAFPC case, a Sigmoidal approximation of the interference is considered in this section with $8$ equally spaced sample points in $s^\mathrm{(dB)}$ between $60$ dB and $200$ dB.

Fig. \ref{fig:avSE_i0} (a) illustrates the average SE versus $i_0$ for IA and non IA methods. It is observed that IA outperforms non IA also in terms of average SE. The reason behind that is that only MT's transmissions that cause strong interference are truncated by $i_0$ while the rest of transmissions are not truncated, allowing for those MTs to reach a good desired power. It can be noticed that there exist an optimal $i_0$ value. 

\begin{figure}[t]
\centering
\includegraphics[width=3in]{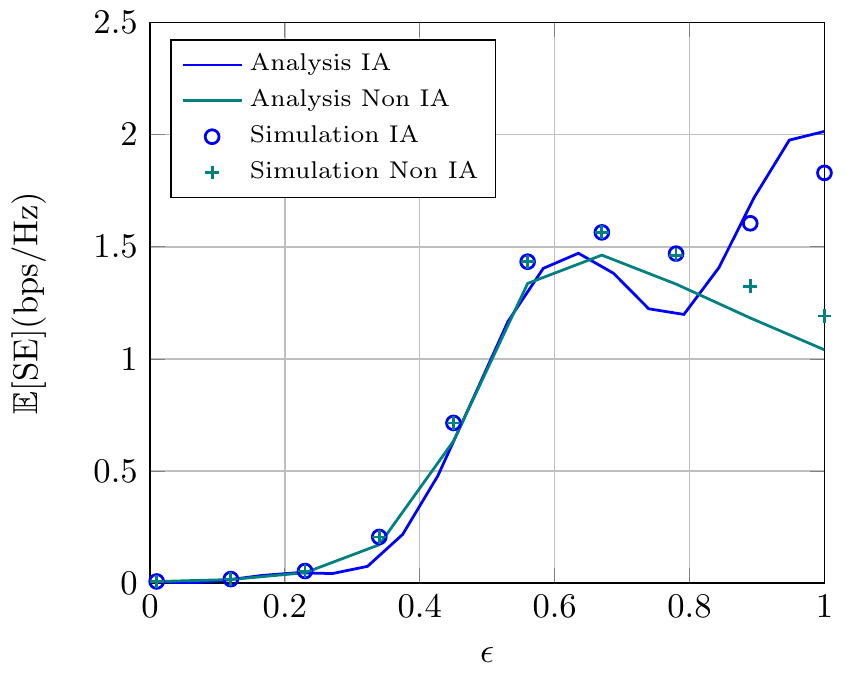}
\caption{Average SE versus $\epsilon$ for IAFPC with $i_0 = 90$ dBm and non IA.}
\label{fig:avSE_epsilon}
\end{figure}

Fig. \ref{fig:avSE_i0} (b) illustrates the average SE versus $p_\mathrm{max}$. It can be observed that reducing $p_\mathrm{max}$ can cause a loss in performance since MTs cannot compensate their path losses and system becomes noise limited. 

Finally, Fig. \ref{fig:avSE_epsilon} illustrates the average SE versus $\epsilon$ for IA and non IA methods. It can be observed that with non IA there exist an optimal value between 0 and 1 that maximizes the SE. This is due to the fact that a total compensation of the path loss and shadowing ($\epsilon=1$) causes strong interference to neighboring cells. Hence a partial compensation obtains a balance between interference and desired power that maximizes the average SE. It is interesting to note that in case of IA the maximum SE is reached with total compensation of path loss. This is due to the fact that IAFPC reaches a good balance between interference and desired power by means of the maximum allowed interference level $i_0$. This parameter offers another degree of freedom to smartly reduce the interference and increase the average SE which has a gain around $33\%$. 


\section{Conclusions}
\label{sec:Conclusion}
In this work a general framework to analyze IAFPC for the UL of HCNs is presented. For the sake of mathematical tractability, a PPP assumption over the positions of interfering MTs need to be considered. Although this brings mathematical tractability, it does not consider existing correlations between BSs and interfering MT positions which leads to inaccurate results. In order to improve the PPP assumption, we propose the use of two indicator functions over the set of interfering positions that add some necessary correlations with the probe BS and also with the interference level $i_0$. By using this framework we obtain a wide variety of performance metrics including average transmitted power, mean and variance of the interference and average SE. 

Aiming to reduce numerical complexity of expressions related to IAFPC analysis, two approaches are considered: (i) approximate the analytic interference term and (ii) perform asymptotic analysis.
The proposed approximations of the interference are : Sigmoidal approximation, and TDA. These approximations lead to considerable reduction on numerical complexity while keeping accurate results. On the other hand, asymptotic analysis allows to obtain the performance of non IA and also to obtain insights for the low $i_0$ regime. 
Finally, with this framework IA and non IA methods have been compared. Performance results show that IAFPC method leads to smaller transmitted power, smaller mean and variance of the interference and 
higher average SE than non IA.
\vspace{-4mm}

\appendices

\section{Proof of Proposition \ref{prop:Soft IAFPC PrXj}} 
\label{app:Soft IAFPC PrXj}
For the case $j \neq m$ the probability of $\mathcal{X}_{\mathrm{MT}_{0}}^{(j,m)}$ can be written as 

\ifOneColumn
\begin{align}
&  \Pr  \left( \mathcal{X}_{\mathrm{MT}_{0}}^{(j,m)} \right)
 =\mathbb{E}_{R_{\mathrm{MT}_{0},(1)}^{(j)},R_{\mathrm{MT}_{0},(2)}^{(j)}} \Bigg[
  \mathbf{1}\left( R_{\mathrm{MT}_{0},(2)}^{(j)}>\left( \frac{t^{(m)}}{t^{(j)}} \right)^{\frac{1}{\alpha }}R_{\mathrm{MT}_{0},(1)}^{(j)} \right) 
\nonumber \\ 
 & \left( \bar{F}_{R_{\mathrm{MT}_{0},(1)}^{(m)}}\left( R_{\mathrm{MT}_{0},(2)}^{(j)} \right)-\bar{F}_{R_{\mathrm{MT}_{0},(1)}^{(m)}}\left( \left( \frac{t^{(m)}}{t^{(j)}} \right)^{\frac{1}{\alpha }}R_{\mathrm{MT}_{0},(1)}^{(j)} \right) \right)\Bigg] 
\end{align}
\else
\begin{align}
&  \Pr  \left( \mathcal{X}_{\mathrm{MT}_{0}}^{(j,m)} \right)
 =\mathbb{E}_{R_{\mathrm{MT}_{0},(1)}^{(j)},R_{\mathrm{MT}_{0},(2)}^{(j)}} \Bigg[
 \nonumber \\
&  \mathbf{1}\left( R_{\mathrm{MT}_{0},(2)}^{(j)}>\left( \frac{t^{(m)}}{t^{(j)}} \right)^{\frac{1}{\alpha }}R_{\mathrm{MT}_{0},(1)}^{(j)} \right) 
\nonumber \\ 
 & \left( \bar{F}_{R_{\mathrm{MT}_{0},(1)}^{(m)}}\left( R_{\mathrm{MT}_{0},(2)}^{(j)} \right)-\bar{F}_{R_{\mathrm{MT}_{0},(1)}^{(m)}}\left( \left( \frac{t^{(m)}}{t^{(j)}} \right)^{\frac{1}{\alpha }}R_{\mathrm{MT}_{0},(1)}^{(j)} \right) \right)\Bigg] 
\end{align}
\fi

where it has been used the fact that $R_{\mathrm{MT}_{0},(1)}^{(m)}$ is independent of $R_{\mathrm{MT}_{0},(1)}^{(j)}$ and $R_{\mathrm{MT}_{0},(2)}^{(j)}$. Then performing expectation using the joint pdf of $R_{\mathrm{MT}_{0},(1)},R_{\mathrm{MT}_{0},(2)}$ given in (\ref{eq:joint pdf R1, R2}) completes the proof. Analogously for the case $j=m$ we have

\ifOneColumn
\begin{align}
\Pr  & \left( \mathcal{X}_{\mathrm{MT}_{0}}^{(j,j)} \right)
 = \mathbb{E}_{R_{\mathrm{MT}_{0},(1)}^{(j)},R_{\mathrm{MT}_{0},(2)}^{(j)}}\Bigg[ 
 \bar{F}_{R_{\mathrm{MT}_{0},(1)}^{(\tilde{j})}}\left( \mathrm{max} \left( \left( \frac{t^{(\tilde{j})}}{t^{(j)}} \right)^{\frac{1}{\alpha }}R_{\mathrm{MT}_{0},(1)}^{(j)},R_{\mathrm{MT}_{0},(2)}^{(j)} \right) \right)\Bigg]  
\end{align}
\else
\begin{align}
\Pr  & \left( \mathcal{X}_{\mathrm{MT}_{0}}^{(j,j)} \right)
 = \mathbb{E}_{R_{\mathrm{MT}_{0},(1)}^{(j)},R_{\mathrm{MT}_{0},(2)}^{(j)}}\Bigg[ 
\nonumber \\ 
& \bar{F}_{R_{\mathrm{MT}_{0},(1)}^{(\tilde{j})}}\left( \mathrm{max} \left( \left( \frac{t^{(\tilde{j})}}{t^{(j)}} \right)^{\frac{1}{\alpha }}R_{\mathrm{MT}_{0},(1)}^{(j)},R_{\mathrm{MT}_{0},(2)}^{(j)} \right) \right)\Bigg]  
\end{align}
\fi

The integral form of the expectation across $R_{\mathrm{MT}_{0},(1)}^{(j)},R_{\mathrm{MT}_{0},(2)}^{(j)}$ of the resulting expression also has primitive which completes the proof.

\section{Proof of Proposition \ref{prop:Soft IAFPC f_RMT_UMT Cond Xjm}} 
\label{app:Soft IAFPC f_RMT_UMT Cond Xjm}
For the case $j \neq m$ the joint pdf of the distances towards the serving BS and most interfered BS can be expressed as

\ifOneColumn
\begin{align}
& f_{R_{\mathrm{MT}_{0}},U_{\mathrm{MT}_{0}}}\left( v,w|\mathcal{X}_{\mathrm{MT}_{0}}^{(j,m)} \right)= 
\frac{\mathrm{d}^{2}}{\mathrm{d}v\mathrm{dw}}\frac{\Pr \left( R_{\mathrm{MT}_{0},(1)}^{(j)}\le v,R_{\mathrm{MT}_{0},(1)}^{(m)}\le w,\mathcal{X}_{\mathrm{MT}_{0}}^{(j,m)} \right)}{\Pr \left( \mathcal{X}_{\mathrm{MT}_{0}}^{(j,m)},\mathcal{A}_{\mathrm{MT}_{0}} \right)} 
\end{align}
\else
\begin{align}
& f_{R_{\mathrm{MT}_{0}},U_{\mathrm{MT}_{0}}}\left( v,w|\mathcal{X}_{\mathrm{MT}_{0}}^{(j,m)} \right)= \nonumber \\ 
& \quad \frac{\mathrm{d}^{2}}{\mathrm{d}v\mathrm{dw}}\frac{\Pr \left( R_{\mathrm{MT}_{0},(1)}^{(j)}\le v,R_{\mathrm{MT}_{0},(1)}^{(m)}\le w,\mathcal{X}_{\mathrm{MT}_{0}}^{(j,m)} \right)}{\Pr \left( \mathcal{X}_{\mathrm{MT}_{0}}^{(j,m)},\mathcal{A}_{\mathrm{MT}_{0}} \right)} 
\end{align}
\fi

The numerator of the previous expression can be obtained as

\ifOneColumn
\begin{align}
& \Pr \left( R_{\mathrm{MT}_{0},(1)}^{(j)}\le v,R_{\mathrm{MT}_{0},(1)}^{(m)}\le w,\mathcal{X}_{\mathrm{MT}_{0}}^{(j,m)} \right) 
=\mathbb{E}_{R_{\mathrm{MT}_{0},(1)}^{(j)},R_{\mathrm{MT}_{0},(2)}^{(j)}}\mathbb{E}_{
R_{\mathrm{MT}_{0},(1)}^{(m)}} \Bigg[ 
\nonumber \\ 
& \quad \mathbf{1}\left( R_{\mathrm{MT}_{0},(1)}^{(m)}>\left( \frac{t^{(m)}}{t^{(j)}} \right)^{
\frac{1}{\alpha }}R_{\mathrm{MT}_{0},(1)}^{(j)} \right) 
\mathbf{1}\left( R_{\mathrm{MT}_{0},(1)}^{(m)}<R_{\mathrm{MT}_{0},(2)}^{(j)} \right) 
\mathbf{1}\left( R_{\mathrm{MT}_{0},(1)}^{(j)}\le v \right)\mathbf{1}\left( R_{\mathrm{MT}_{0},(1)}^{(m)}\le w \right)\Bigg] 
\nonumber \\ 
& \quad =\int\limits_{r_{1}^{(m)}=0}^{w}{\int\limits_{r_{1}^{(j)}=0}^{v}{\int\limits_{r_{2}^{(j)}=_{1}^{(j)}}^{\infty }{f_{R_{\mathrm{MT}_{0},(1)}^{(j)},R_{\mathrm{MT}_{0},(2)}^{(j)}}\left( r_{1}^{(j)},r_{2}^{(j)} \right)}}} 
f_{R_{\mathrm{MT}_{0},(1)}^{(m)}}\left( r_{1}^{(m)} \right)\mathbf{1}\left( R_{\mathrm{MT}_{0},(1)}^{(m)}<R_{\mathrm{MT}_{0},(2)}^{(j)} \right) 
\nonumber  \\ 
& \quad \mathbf{1}\left( r_{2}^{(j)}>\left( \frac{t^{(m)}}{t^{(j)}} \right)^{\frac{1}{\alpha }}r_{1}^{(j)} \right)\mathrm{d}r_{2}^{(j)}\mathrm{d}r_{1}^{(j)}\mathrm{d}r_{1}^{(m)}  
\end{align}
\else
\begin{align}
& \Pr \left( R_{\mathrm{MT}_{0},(1)}^{(j)}\le v,R_{\mathrm{MT}_{0},(1)}^{(m)}\le w,\mathcal{X}_{\mathrm{MT}_{0}}^{(j,m)} \right) 
\nonumber \\ 
& \quad =\mathbb{E}_{R_{\mathrm{MT}_{0},(1)}^{(j)},R_{\mathrm{MT}_{0},(2)}^{(j)}}\mathbb{E}_{R_{\mathrm{MT}_{0},(1)}^{(m)}} \Bigg[ 
\nonumber \\ 
& \quad \mathbf{1}\left( R_{\mathrm{MT}_{0},(1)}^{(m)}>\left( \frac{t^{(m)}}{t^{(j)}} \right)^{\frac{1}{\alpha }}R_{\mathrm{MT}_{0},(1)}^{(j)} \right) 
\nonumber \\ 
& \quad \mathbf{1}\left( R_{\mathrm{MT}_{0},(1)}^{(m)}<R_{\mathrm{MT}_{0},(2)}^{(j)} \right) 
\nonumber \\ 
& \quad \mathbf{1}\left( R_{\mathrm{MT}_{0},(1)}^{(j)}\le v \right)\mathbf{1}\left( R_{\mathrm{MT}_{0},(1)}^{(m)}\le w \right)\Bigg] 
\nonumber \\ 
& \quad =\int\limits_{r_{1}^{(m)}=0}^{w}{\int\limits_{r_{1}^{(j)}=0}^{v}{\int\limits_{r_{2}^{(j)}=_{1}^{(j)}}^{\infty }{f_{R_{\mathrm{MT}_{0},(1)}^{(j)},R_{\mathrm{MT}_{0},(2)}^{(j)}}\left( r_{1}^{(j)},r_{2}^{(j)} \right)}}} 
\nonumber \\ 
& \quad f_{R_{\mathrm{MT}_{0},(1)}^{(m)}}\left( r_{1}^{(m)} \right)\mathbf{1}\left( R_{\mathrm{MT}_{0},(1)}^{(m)}<R_{\mathrm{MT}_{0},(2)}^{(j)} \right) 
\nonumber  \\ 
& \quad \mathbf{1}\left( r_{2}^{(j)}>\left( \frac{t^{(m)}}{t^{(j)}} \right)^{\frac{1}{\alpha }}r_{1}^{(j)} \right)\mathrm{d}r_{2}^{(j)}\mathrm{d}r_{1}^{(j)}\mathrm{d}r_{1}^{(m)}  
\end{align}
\fi

Performing the inner integral with $r^{(j)}_{2}$ and then applying Leibniz integration rule completes the proof. Analogously for the case $j=m$ the numerator can be expressed as

\ifOneColumn
\begin{align}
& \int\limits_{r_{1}^{(j)}=0}^{v}{\int\limits_{r_{2}^{(j)}=_{1}^{(j)}}^{w}{f_{R_{\mathrm{MT}_{0},(1)}^{(j)},R_{\mathrm{MT}_{0},(2)}^{(j)}}\left( r_{1}^{(j)},r_{2}^{(j)} \right)}} 
\bar{F}_{R_{\mathrm{MT}_{0},(1)}^{(\tilde{j})}}\left( \mathrm{max} \left( \left( \frac{t^{(\tilde{j})}}{t^{(j)}} \right)^{\frac{1}{\alpha }}r_{1}^{(j)},r_{2}^{(j)} \right) \right)\mathrm{d}r_{2}^{(j)}\mathrm{d}r_{1}^{(j)}  
\end{align}
\else
\begin{align}
& \int\limits_{r_{1}^{(j)}=0}^{v}{\int\limits_{r_{2}^{(j)}=_{1}^{(j)}}^{w}{f_{R_{\mathrm{MT}_{0},(1)}^{(j)},R_{\mathrm{MT}_{0},(2)}^{(j)}}\left( r_{1}^{(j)},r_{2}^{(j)} \right)}} 
\nonumber \\ 
& \quad \bar{F}_{R_{\mathrm{MT}_{0},(1)}^{(\tilde{j})}}\left( \mathrm{max} \left( \left( \frac{t^{(\tilde{j})}}{t^{(j)}} \right)^{\frac{1}{\alpha }}r_{1}^{(j)},r_{2}^{(j)} \right) \right)\mathrm{d}r_{2}^{(j)}\mathrm{d}r_{1}^{(j)}  
\end{align}
\fi

Finally applying Leibniz integration rule completes the proof.

\section{Proof of Proposition \ref{prop:Soft IAFPC LI}} 
\label{app:Soft IAFPC LI}
The Laplace transform of the interference can be expressed as

\ifOneColumn
\begin{align}
\mathcal{L}_{I} & \left( s|\mathcal{X}_{\mathrm{MT}_{0}}^{(j)} \right)=\mathbb{E}_{I}\left[ \mathrm{e}^{-sI}|\mathcal{X}_{\mathrm{MT}_{0}}^{(j)} \right]
=\prod\limits_{k\in \mathcal{K}}{\mathbb{E}_{\Psi ^{(k)}}} 
 \prod\limits_{\mathrm{MT}_{i}\in \Psi ^{(k)}}{\mathbb{E}_{R_{\mathrm{MT}_{i}},U_{\mathrm{MT}_{i}}}}
\Bigg[\mathbb{E}_{H_{\mathrm{MT}_{i}}}\exp \Bigg(-sH_{\mathrm{MT}_{i}}\left( \tau D_{\mathrm{MT}_{i}} \right)^{-\alpha } 
\nonumber \\ 
& p_{\mathrm{MT}}\left( R_{\mathrm{MT}_{0}},U_{\mathrm{MT}_{0}} \right)\mathbf{1}\left( \mathcal{O}_{\mathrm{MT}_{i}}^{(j,k)} \right)\mathbf{1}\left( \mathcal{Z}_{\mathrm{MT}_{i}} \right) \Bigg)|\mathcal{X}_{\mathrm{MT}_{i}}^{(k)}\Bigg] 
\end{align}
\else
\begin{align}
\mathcal{L}_{I} & \left( s|\mathcal{X}_{\mathrm{MT}_{0}}^{(j)} \right)=\mathbb{E}_{I}\left[ \mathrm{e}^{-sI}|\mathcal{X}_{\mathrm{MT}_{0}}^{(j)} \right]
=\prod\limits_{k\in \mathcal{K}}{\mathbb{E}_{\Psi ^{(k)}}} 
\nonumber \\ 
& \prod\limits_{\mathrm{MT}_{i}\in \Psi ^{(k)}}{\mathbb{E}_{R_{\mathrm{MT}_{i}},U_{\mathrm{MT}_{i}}}}
\Bigg[\mathbb{E}_{H_{\mathrm{MT}_{i}}}\exp \Bigg(-sH_{\mathrm{MT}_{i}}\left( \tau D_{\mathrm{MT}_{i}} \right)^{-\alpha } 
\nonumber \\ 
& p_{\mathrm{MT}}\left( R_{\mathrm{MT}_{0}},U_{\mathrm{MT}_{0}} \right)\mathbf{1}\left( \mathcal{O}_{\mathrm{MT}_{i}}^{(j,k)} \right)\mathbf{1}\left( \mathcal{Z}_{\mathrm{MT}_{i}} \right) \Bigg)|\mathcal{X}_{\mathrm{MT}_{i}}^{(k)}\Bigg] 
\end{align}
\fi

Applying the PGFL, performing expectation over the fading and conditioning over the event $\mathcal{Q}_{\mathrm{MT}_{i}}^{(n)}$ yields 

\ifOneColumn
\begin{align}
\mathcal{L}_{I} & \left( s|\mathcal{X}_{\mathrm{MT}_{0}}^{(j)} \right) =
 \exp \Bigg(-\sum\limits_{k\in \mathcal{K}}{2\pi \lambda ^{(k)}}\sum\limits_{n\in \mathcal{K}}{\Pr \left( \mathcal{Q}_{\mathrm{MT}_{i}}^{(n)}|\mathcal{X}_{\mathrm{MT}_{i}}^{(k)} \right)\times } 
 \int\limits_{r=0}^{\infty }{\int\limits_{u=\left( \frac{t^{(n)}}{t^{(k)}} \right)^{\frac{1}{\alpha }}r}^{\infty }{f_{R_{\mathrm{MT}_{i}},U_{\mathrm{MT}_{i}}}\left( r,u|\mathcal{X}_{\mathrm{MT}_{i}}^{(k,n)} \right)}} 
\nonumber \\ 
& \int\limits_{\rho = 
\mathrm{max} \left( \left( \frac{t^{(j)}}{t^{(k)}} \right)^{\frac{1}{\alpha }}r,\frac{1}{\tau }\left( \frac{p_{\mathrm{MT}}\left( r,u \right)}{i_{0}} \right)^{\frac{1}{\alpha }} \right)
}
^{\infty }\frac{s\left( \tau \rho  \right)^{-\alpha }p_{\mathrm{MT}}\left( r,u \right)
\rho \mathrm{d}\rho \mathrm{d}r}
{1+s\left( \tau \rho  \right)^{-\alpha }p_{\mathrm{MT}}\left( r,u \right)} \Bigg)
\end{align}
\else
\begin{align}
\mathcal{L}_{I} & \left( s|\mathcal{X}_{\mathrm{MT}_{0}}^{(j)} \right) =
 \exp \Bigg(-\sum\limits_{k\in \mathcal{K}}{2\pi \lambda ^{(k)}}\sum\limits_{n\in \mathcal{K}}{\Pr \left( \mathcal{Q}_{\mathrm{MT}_{i}}^{(n)}|\mathcal{X}_{\mathrm{MT}_{i}}^{(k)} \right)\times } 
\nonumber \\ 
& \int\limits_{r=0}^{\infty }{\int\limits_{u=\left( \frac{t^{(n)}}{t^{(k)}} \right)^{\frac{1}{\alpha }}r}^{\infty }{f_{R_{\mathrm{MT}_{i}},U_{\mathrm{MT}_{i}}}\left( r,u|\mathcal{X}_{\mathrm{MT}_{i}}^{(k,n)} \right)}} 
\nonumber \\ 
& \int\limits_{\rho = 
\mathrm{max} \left( \left( \frac{t^{(j)}}{t^{(k)}} \right)^{\frac{1}{\alpha }}r,\frac{1}{\tau }\left( \frac{p_{\mathrm{MT}}\left( r,u \right)}{i_{0}} \right)^{\frac{1}{\alpha }} \right)
}
^{\infty }\frac{s\left( \tau \rho  \right)^{-\alpha }p_{\mathrm{MT}}\left( r,u \right)
\rho \mathrm{d}\rho \mathrm{d}r}
{1+s\left( \tau \rho  \right)^{-\alpha }p_{\mathrm{MT}}\left( r,u \right)} \Bigg)
\end{align}
\fi

Finally integrating the inner integral over $\rho$ completes the proof.

\section*{Acknowledgment}
This work has been partly supported by the Spanish Government under grant TEC2013-44442-P,  the University of Malaga and by the European Commission under the auspices of the H2020-MSCA-ITN-2014 5Gwireless project (grant 641985).

\ifCLASSOPTIONcaptionsoff
  \newpage
\fi


\bibliographystyle{IEEEtran}
\bibliography{UL_Power_Control}

%








\end{document}